\documentclass[journal]{IEEEtran}

\usepackage[english]{babel}

\usepackage{ifpdf}

\usepackage{cite} 
\usepackage{url}
\usepackage{hyperref}

\ifCLASSINFOpdf
	\usepackage[pdftex]{graphicx}
	\graphicspath{{./figures/}}
 	\DeclareGraphicsExtensions{.pdf,.jpeg,.png}
\else
	\usepackage[dvips]{graphicx}
	\graphicspath{./figures/}
\fi
\usepackage{color}
\usepackage{pgf, tikz, pgfplots}
\usetikzlibrary{shapes, arrows, automata}
\usetikzlibrary{calc,hobby,decorations}

\usepackage[cmex10]{amsmath}
\usepackage{amsfonts, amssymb, amsthm}
\usepackage{mathrsfs}
\usepackage[mathscr]{euscript}

\newcommand{\vertiii}[1]{{\left\vert\kern-0.25ex\left\vert\kern-0.25ex\left\vert #1 
    \right\vert\kern-0.25ex\right\vert\kern-0.25ex\right\vert}}

\usepackage{array}
\usepackage{enumerate}
\usepackage{multirow}
\usepackage{rotating}
\usepackage{subcaption}
\usepackage{needspace}

\input{mySymbol.sty}
\input{pennColors.sty}

\def\Tr{\mathsf{T}}
\def\Hr{\mathsf{H}}

\renewcommand{\blue}{\color{black}}

\newtheorem{assumption}{\hspace{0pt}\bf AS\hspace{-0.15cm}}
\newtheorem{lemma}{\hspace{0pt}\bf Lemma}
\newtheorem{proposition}{\hspace{0pt}\bf Proposition}

\newtheorem{theorem}{\hspace{0pt}\bf Theorem}

\newtheorem{remark}{\hspace{0pt}\bf Remark}

\newtheorem{definition}{\hspace{0pt}\bf Definition}

\begin{document}

\title{Gated Graph Recurrent Neural Networks}

\author{Luana~Ruiz, Fernando~Gama~
        and~Alejandro~Ribeiro
\thanks{The work in this paper was supported by NSF CCF 1717120, ARO W911NF1710438, ARL DCIST CRA W911NF-17-2-0181, ISTC-WAS and Intel DevCloud. Preliminary results appear at the EUSIPCO19 conference \cite{ruiz2019gated}. L. Ruiz, F. Gama and A. Ribeiro are with the Dept. of Electrical and Systems Eng., Univ. of Pennsylvania.  Email: \{rubruiz,fgama,aribeiro\}@seas.upenn.edu.
}
}

\markboth{IEEE TRANSACTIONS ON SIGNAL PROCESSING (SUBMITTED)}%
{Gated Graph Recurrent Neural Networks}

\maketitle

\begin{abstract}
    Graph processes exhibit a temporal structure determined by the sequence index and and a spatial structure determined by the graph support. To learn from graph processes, an information processing architecture must then be able to exploit both underlying structures. We introduce Graph Recurrent Neural Networks (GRNNs) \blue{as a general learning framework that achieves} this goal by leveraging \blue{the notion of a recurrent hidden state} together with graph signal processing (GSP). In the GRNN, the number of learnable parameters is independent of the length of the sequence and of the size of the graph, guaranteeing scalability. We prove that GRNNs are permutation equivariant and that they are stable to perturbations of the underlying graph support. \blue{To address the problem of vanishing gradients, we also put forward gated GRNNs with three different gating mechanisms: time, node and edge gates. In numerical experiments involving both synthetic and real datasets, time-gated GRNNs are shown to improve upon GRNNs in problems with long term dependencies, while node and edge gates help encode long range dependencies present in the graph.
The numerical results also show that GRNNs outperform GNNs and RNNs, highlighting the importance of taking both the temporal and graph structures of a graph process into account.}
\end{abstract}


\IEEEpeerreviewmaketitle


\section{Introduction} \label{sec:intro}



\blue{Graph neural networks (GNNs) \cite{Kipf17-ClassifGCN, Defferrard17-CNNGraphs, Gama19-Architectures} are a popular information processing architecture in graph signal processing (GSP), having found applications in problems such as recommender systems \cite{Ruiz20-Nonlinear} and robot path planning \cite{tolstaya2020learning}. Their popularity is largely explained by state-of-the-art performances achieved in several learning tasks \cite{Goodfellow16-DeepLearning} which, in turn, are related to invariance and stability properties that they inherit from graph convolutions \cite{Gama19-Stability}. At the same time, GNNs are somewhat limited in that they are designed to process data with only one type of structure---the graph. This creates a gap for an important class of GSP problems involving graph signals that can also change with time, which we call \emph{graph processes}. 

Graph processes contain a time dimension, reflected by the indices of the sequence, and a \textit{fixed} graph structure, which is inherent to graph signals {\cite{girault2015translation,marques2017stationary}}. They have been used to model data such as weather variables on weather station networks \cite{perraudin2017stationary} and seismic wave readings on a network of seismographs \cite{grassi2017time}. In problems involving sequences where the data elements are Euclidean, recurrent neural networks (RNNs) tend to be the architecture of choice, as they leverage recurrence to model the time dependencies present in sequential data \cite{pascanu2013construct,graves2013generating, schuster1997bidirectional}. For graph processes, this motivates retaining the same recurrence relation but replacing the input-to-state and state-to-state linear transformations of the RNN by linear graph filters, which allows taking both the temporal and the graph structure of graph processes into account.

Implementations of such a \emph{graph recurrent} architecture can be found in the literature, usually focusing on either a particular type of graph \cite{seo2018structured} (undirected) or problem \cite{li2017diffusion} (traffic forecasting). 
However, the basic construction of a graph recurrent neural network (GRNN) and its fundamental properties have not yet been investigated in detail. In this paper, we propose to do so by using GSP to come up with a unified framework for GRNNs.  
We write graph operations in terms of a generic graph shift operator (GSO), making for a more general architecture and drawing attention away from specific matrix representations, e.g. the random walk matrix used to define the architecture in \cite{li2017diffusion} or the graph Laplacian used to define the architecture in \cite{seo2018structured}. We also introduce a convolutional parametrization of the input-to-state and state-to-state operations in Section \ref{sec:GRNN} which, while not novel, guarantees equivariance to node relabelings and makes the number of parameters independent of the size of the graph \cite{Ruiz20-Nonlinear}. In Section \ref{sec:stability}, one of our main contributions is then proving that, like GNNs, GRNNs also exhibit stability to relative perturbations of the underlying graph. This means that changes in the output caused by changes in the graph are bounded by the size of the perturbation \cite{Gama19-Stability}, with the key difference that in GRNNs this stability deteriorates with the length of the sequence [cf. Theorems \ref{thm:stability} and \ref{thm:GGRNNstability}]. Another important contribution is the introduction of generic input and forget gate operators that we break down in three gating strategies---time, node and edge gating---, all of which interact with the graph in different ways (Section \ref{sec:gatedGRNN}). 
In the numerical experiments in Section \ref{sec:sims}, time gates prove useful for encoding long term temporal dependencies, while node and edge gates help encode long range spatial dependencies on the graph. Note that these contributions have broad applicability as they extend to all architectures fitting the GRNN framework, e.g. the aforementioned  \cite{seo2018structured,li2017diffusion}.
}

\blue{Related work on learning problems involving graph processes also includes \cite{zhang2018gaan} and \cite{yu2017spatio}, which, like \cite{li2017diffusion}, have an emphasis on traffic forecasting. The gated attention networks (GaANs) from \cite{zhang2018gaan} replace the linear transformations of the RNN by graph attention networks (GANs) \cite{Velickovic18-GraphAttentionNetworks}, making for an architecture that is not convolutional and hence has different properties than GRNNs. Meanwhile, the architecture introduced in \cite{yu2017spatio} stacks GNNs and gated CNNs to learn spatiotemporal dependencies and is therefore not recurrent.
}
Other somewhat related works include the gated graph sequence neural networks \cite{li2015gated} and the recurrent formulation in \cite{ioannidis2018recurrent}. The architecture in \cite{li2015gated} learns sequential representations from graphs, but not from graph signals or processes. This is a fundamental difference since in learning from graphs the graph is seen as data, while in learning from graph signals the graph is given (i.e., a hyperparameter of the learning architecture). The work in \cite{ioannidis2018recurrent} uses recurrence as a means of re-introducing the input at every layer to capture multiple types of diffusion, but does not consider data consisting of temporal sequences.
\blue{We point out that neural network architectures designed to predict dynamic graphs {\cite{wu2020evonet}} or to process signals supported on them {\cite{tolstaya19-flocking,Li20-Planning}} are tangential to our work, as, coherent with the GSP working assumption for graph processes {\cite{girault2015translation,marques2017stationary}}, we only consider fixed graphs.}
Finally, while what we call \emph{time} and \emph{node gating} strategies have been used in \cite{li2017diffusion, zhang2018gaan, yu2017spatio, seo2018structured, li2015gated}, \blue{we not only introduce \emph{edge gating} but also provide an interpretation of all three gating strategies with respect to the underlying graph.}

The remaining sections of this paper are organized as follows. Section \ref{sec:RNN} goes over RNNs and graph signal processing. Following the introduction of the GRNN framework in Section \ref{sec:GRNN}, in Section \ref{sec:stability} we analyze GRNN stability and, in Section \ref{sec:gatedGRNN}, introduce Gated GRNNs. In Section \ref{sec:sims}, we evaluate the performance of all GRNN architectures, \blue{as well as the architectures from \cite{seo2018structured,li2017diffusion},} in a synthetic $k$-step prediction experiment and in \blue{three real-world experiments: earthquake epicenter estimation, traffic forecasting and epidemic tracking.} Concluding remarks are presented in Section \ref{sec:conclusions}.


\section{Recurrent Neural Networks and Graph Data} \label{sec:RNN} 



Instrumental to the design of GRNNs are the concept of traditional RNNs as models of sequential data and the theory of graph signal processing. \blue{We thus discuss how RNNs can be used to process sequential data} (Section~\ref{subsec:RNN}), and follow with an overview of graph signals, graph convolutions and graph processes (Section~\ref{subsec:graphData}). \blue{Note that the bias terms are omitted in Section \ref{subsec:RNN} to unburden the notation.}


\subsection{Recurrent Neural Networks (RNNs)} \label{subsec:RNN}

\blue{
    Let $\{\bbx_{t}\}_{t \in \mbN_{0}}$ be a sequence of $N$-dimensional data points $\bbx_{t} \in \reals^{N}$. A \emph{recurrent neural network} (RNN) learns to extract information from this sequence in the form of a \emph{hidden state} variable $\bbz_{t} \in \reals^{N}$. The states $\bbz_{t}$ are learned from the sequence $\{\bbx_t\}_{t \in \mbN_0}$ using a nonlinear map that takes the current data point $\bbx_{t}$ and the previous hidden state $\bbz_{t-1}$ as inputs, and outputs the updated hidden state $\bbz_{t}$. This map is parametrized as
\begin{equation} \label{eqn:RNNhidden}
\bbz_{t} = \sigma \left( \bbA \bbx_{t} + \bbB \bbz_{t-1} \right)
\end{equation}
where $\bbA \in \reals^{N \times N}$ and $\bbB \in \reals^{N \times N}$ are linear operators and $\sigma: \reals \to \reals$ is a pointwise nonlinearity, i.e. $[\sigma(\bbx)]_{i} = \sigma([\bbx]_{i})$, see \cite[Fig. 10.13]{Goodfellow16-DeepLearning} for a computational graph. We point out that, while it is not necessary for $\bbz_{t}$ and $\bbx_{t}$ to share the same dimensions (in most realizations of RNNs, they often do not), we assume so here for ease of exposition.
}

The sequence $\{\bbx_{t}\}$ is typically accompanied by a target representation $\ccalY$, which can be seen as a more appropriate representation of $\{\bbx_{t}\}$ for the task at hand. The elements of $\ccalY$ could be, e.g., a single value $\bby \in \ccalY$ to summarize information from the entire sequence, like a sentiment describing a tweet \cite{baziotis2017datastories}; or, they could be another sequence $\{\bby_{t}\}_{t \in \mbN_{0}}$, $\bby_{t} \in \reals^{M}$, which is the case in automatic speech recognition \cite{miao2015eesen}.
RNNs estimate $\ccalY$ by applying a second nonlinear map, $\bbPhi: \reals^{N} \to \reals^{M}$, to the hidden state. When the target representation is a sequence, this map is parametrized as
\begin{equation} \label{eqn:RNNoutput}
\bbPhi(\bbz_{t}) = \rho\left( \bbC \bbz_{t} \right)
\end{equation}
where $\bbC \in \reals^{M \times N}$ is the linear \emph{output} map and $\rho:\reals \to \reals$ is the pointwise nonlinearity used to compute the output, $[\rho(\bbx)]_{i} = \rho([\bbx]_{i})$. In cases where a single output value $\bby$ is associated with the sequence $\{\bbx_t\}$, we can estimate $\bby$ from the state at the end $T$ of the sequence, $\bbPhi(\bbz_{T}) = \rho\left( \bbC \bbz_{T} \right)$.

Given a training set $\{(\{\bbx_{t}\}, \ccalY)\}$ comprised of several sequences $\{\bbx_{t}\}$ and their associated representations $\ccalY$, the optimal linear maps $\bbA$, $\bbB$ and $\bbC$ are obtained by minimizing some loss function $\ccalL(\bbPhi(\bbz_{t}), \ccalY)$ (or $\ccalL(\bbPhi(\bbz_{T}), \ccalY)$) over the training set. This learning framework makes the hidden state adaptable to the task at hand, exploiting the available training examples to determine which pieces of sequential information are relevant to store in the hidden state $\bbz_{t}$. 

Key to the success of RNNs is the fact that the number of parameters (entries) in the linear operators $\bbA$, $\bbB$ and $\bbC$ \emph{do not} depend on the time index $t$. In other words, the same linear operators are applied throughout the entire sequence. This parameter-sharing scheme across the time-dimension has two main advantages: it keeps the number of parameters under control and, simultaneously, allows learning from sequences of variable length. This is consistent with \blue{our recurrent approximation model for the hidden state}, where each \emph{learned} state only depends on the current input and on the previous state. Regardless of the start time $t_0$, as long as the current value of the input and of the previous state are the same, the updated state will always be the same. 


\subsection{Graph data} \label{subsec:graphData}

In upcoming sections, our focus will be on adapting RNNs to process \emph{graph data}. Let $\ccalG = (\ccalV, \ccalE, \ccalW)$ be a graph, where $\ccalV = \{1,\ldots,N\}$ is the set of nodes, $\ccalE \subseteq \ccalN \times \ccalN$ is the set of edges and $\ccalW: \ccalE \to \reals$ is a weight function assigning proximity weights to the edges in $\ccalE$. We say that a data sample $\bbx$ is a \emph{graph signal} \cite{Sandryhaila13-DSPG, Shuman13-SPG} if its entries are related through the graph $\ccalG$. To be precise, each node $n \in \ccalV$ is assigned an entry $[\bbx]_{n} = x_{n}$, and the entries $x_{i}$ and $x_{j}$ are presumed related if there is an edge between them. The \emph{strength} of this relationship is usually measured by the edge weight. 

To provide a better interpretation of the relationship between the signal $\bbx$ and the graph $\ccalG$, we define the graph shift operator (GSO) $\bbS \in \reals^{N \times N}$ as a matrix that encodes the sparsity pattern of $\ccalG$ by requiring that entries $[\bbS]_{ij} = s_{ij}$ be nonzero only if $i=j$ or $(j,i) \in \ccalE$. The value $s_{ij}$ reflects the influence that the components of the signal at nodes $i$ and $j$ exert on one another. Examples of GSOs include the adjacency matrix \cite{Sandryhaila13-DSPG}, the Laplacian matrix \cite{Shuman13-SPG}, the random walk matrix \cite{Heimowitz17-MarkovGSP} and their normalizations. By construction, the GSO can be used to define $\bbS \bbx$ as the elementary linear and local map between graph signals. We say that this map is local because the $i$th entry of the output, $[\bbS \bbx]_{i}$, is a linear combination of the signal components in the one-hop neighborhood of $i$. Explicitly,
\begin{equation} \label{eqn:graphShift}
    [\bbS \bbx]_{i}
        = \sum_{j = 1}^{N} [\bbS]_{ij} [\bbx]_{j} 
        = \sum_{j \in \ccalN_{i}} s_{ij} x_{j} .
\end{equation}
where $\ccalN_{i} = \{j \in \ccalN: (j,i) \in \ccalE\}$ denotes the set of immediate neighbors of $i$.
The second equality follows from the fact that $[\bbS]_{ij} = 0$ for all $j \notin \ccalN_{i}$. In a sense, we can view the operation $\bbS \bbx$ as a shift (or diffusion) of the signal on the graph, where the value of the signal at each node is updated as a linear combination of signal values at neighboring nodes. 

The notion of \emph{graph shifts} can be used to define \emph{graph convolutions} analogous to \emph{time} convolutions. Formally, we define the graph convolution as a weighted sum of shifted versions of the signal \cite{Gama19-GraphConv, du2018graph},
\begin{equation} \label{eqn:graphConv}
    \bbA(\bbS) \bbx = \sum_{k=0}^{K-1} a_{k} \bbS^{k} \bbx.
\end{equation}
Note that, since $\bbS^{k} \bbx = \bbS (\bbS^{k-1} \bbx)$, repeated applications of the linear map $\bbS$ entail successive exchanges with neighboring nodes; this means that the graph convolution can be computed as a series of local operations. Also note that the operation $\bbS^{k} \bbx$  produces a summary of the information contained in the $k$-hop neighborhood of each node. For $0 \leq k \leq K-1$, the filter \emph{coefficients} (or filter \emph{taps}) $\bba = [a_{0},a_{1},\ldots,a_{K-1}] \in \reals^{K}$ assign different importances to the information located in each $k$-hop neighborhood. Following the graph signal processing terminology, $\bbA(\bbS) \in \reals^{N \times N}$ is called a \emph{linear shift-invariant graph filter} (LSI-GF) \cite{Segarra17-Linear}, which reinforces the analogy with time-invariant filters and the convolution operation.

In what follows, the data sequences that we consider will be \emph{graph processes}. A graph process is a sequence $\{\bbx_{t}\}_{t \in \mbN_{0}}$ of signals $\bbx_t \in \reals^N$ supported on the graph $\ccalG$.  Alternatively, a graph process can also be seen as a time-varying graph signal where the values of the signal at each node change over time \cite{Gama19-GLLN, Gama19-Control}. While traditional RNNs (Section~\ref{subsec:RNN}) successfully exploit the sequential structure of data, they fail to account for other structures that may be present in $\bbx_t$; however, as substantiated by the remarkable performance achieved by CNNs \cite{lecun15-deeplearning, kuo17-recos, yosinski2014transferable} and GNNs \cite{Kipf17-ClassifGCN, Defferrard17-CNNGraphs, Gama19-Architectures}, exploiting the data's \emph{spatial} structure is of paramount importance. 


\section{Graph Recurrent Neural Networks} \label{sec:GRNN}



%

\begin{figure*}[t]
	\centering
	\includegraphics[width=0.86\textwidth]{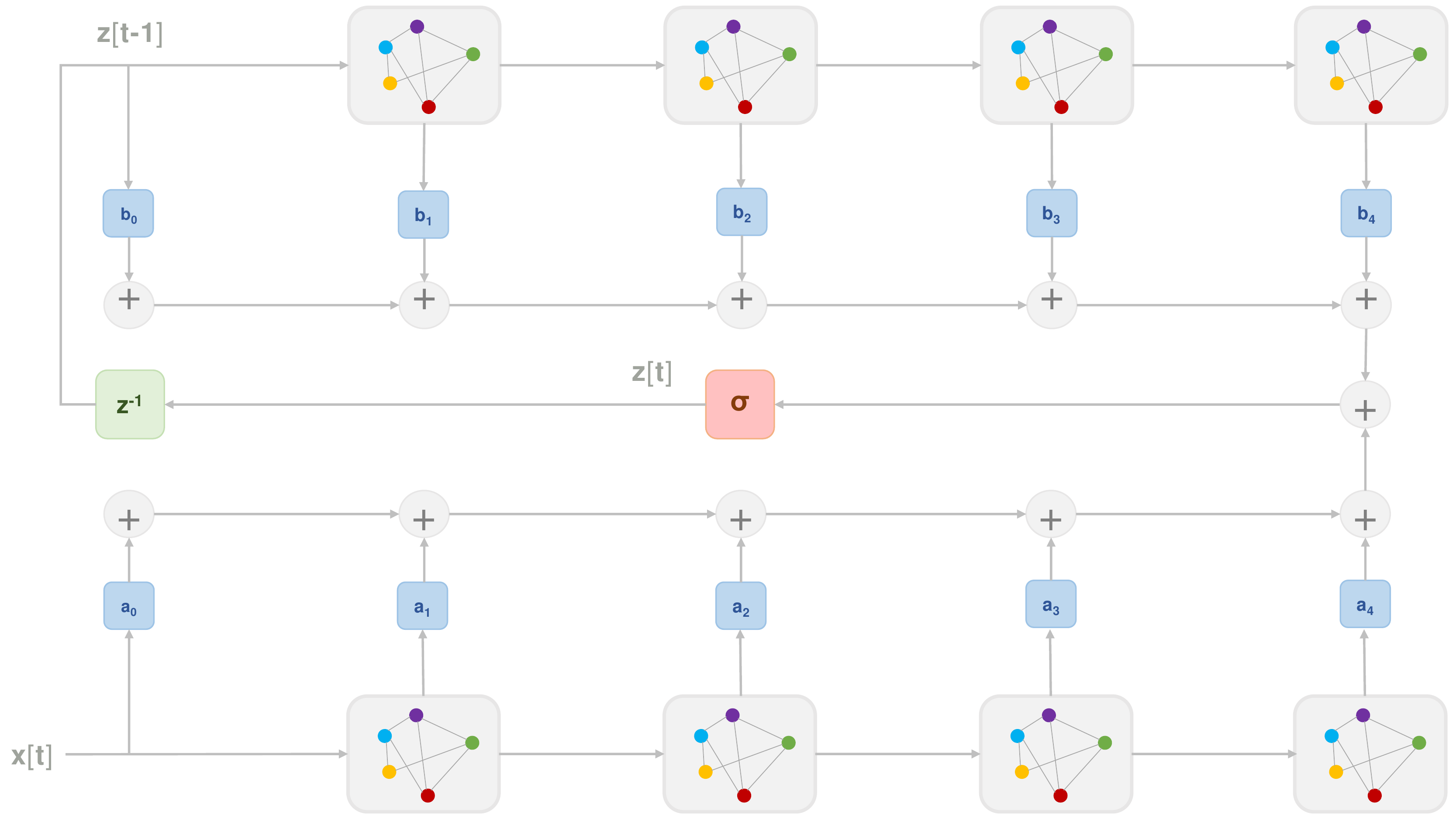} \vspace{1em}
	\caption{State computation in a graph recurrent neural network with $K=5$. Gray blocks with graphs on the inside stand for graph shifts, blue blocks for linear weights, the red block for a pointwise nonlinearity and the green block for a time delay. } 
	\label{fig:grnn}
\end{figure*}	

As previously noted in Section~\ref{subsec:RNN}, RNNs \blue{are systems which exploit recurrence to learn} dependencies in sequences of variable length with a number of parameters that is independent of time. However, the number of parameters still depends on the dimension $N$, which not only prevents RNNs from scaling to inputs with large dimensions but, more importantly, hinders their ability to account for other structures inherent to the data. Accounting for structure is desirable first because when we parametrize operations in terms of the structure of the data, we are effectively adding a constraint to the optimization problem, shrinking the feasible set and making it easier to find close-to-optimal solutions; and second, because it allows us to leverage repeating and/or symmetrical motifs in the data to extract shared features and simplify model parametrization.
In the case of graph processes, we will thus adapt the operations performed by RNNs to take the graph structure into account. We assume that the state $\bbz_{t} \in \reals^{N}$ is itself a graph signal, so that each entry $[\bbz_{t}]_{n}$ is a \emph{nodal} hidden state. The updated state can then be calculated by parametrizing the linear maps $\bbA$ and $\bbB$ by the graph shift operator $\bbS$, yielding
\begin{equation} \label{eqn:GRNNhidden}
    \bbz_{t} = \sigma \big( \bbA(\bbS) \bbx_{t} + \bbB(\bbS) \bbz_{t-1} \big)
\end{equation}
\blue{where we have omitted the bias term to unburden the notation.}
We call this generic architecture the \emph{graph recurrent neural network} (GRNN). \blue{Note that the computational graph for obtaining \eqref{eqn:GRNNhidden} is analogous to the one in \cite[Fig. 10.13]{Goodfellow16-DeepLearning} but replacing the linear transforms by graph filters.} Although $\bbA(\bbS)$ and $\bbB(\bbS)$ can be arbitrary functions of $\bbS$ \cite{isufi19-edgenets}, we opt for the graph convolution, \eqref{eqn:graphConv} so that there are only $K$ parameters to learn for each filter ($\bba= [a_{0},\ldots,a_{K-1}] \in \reals^K$ and $\bbb = [b_{0},\ldots,b_{K-1}] \in \reals^K$). This allows computations to be done locally and, like in GNNs, ensures that the number of parameters is independent of the size of the graph. Other advantages of graph convolutions are that they are permutation equivariant and stable to graph perturbations \cite{Gama19-Stability}, a fact that we use to derive a stability result for GRNNs in Section \ref{sec:stability}.

To estimate the target representation $\ccalY$, we can once again leverage the fact that the hidden state $\bbz_{t}$ is a graph signal and use a GNN $\bbPhi(\bbz_{t};\bbS)$ \cite{Gama19-Architectures} to compute the estimate $\hat{\ccalY}$. In cases where $\bby_{t}$ is itself a graph signal (e.g., in regression or forecasting), $\bbPhi$ can be a simple one-layer graph filter followed by an activation function $\rho$,
\begin{equation} \label{eqn:GRNNoutput}
    \hby_{t} = \rho \left( \bbC(\bbS) \bbz_{t} \right)
\end{equation}
where we parametrize the filter $\bbC(\bbS)$ as a graph convolution \eqref{eqn:graphConv} with $K$ filter taps ($\bbc \in \reals^K$) to make sure that the number of parameters of the architecture \eqref{eqn:GRNNhidden}-\eqref{eqn:GRNNoutput} is independent of the size of the graph. 
If $\bby_t$ has dimension $M \neq N$ then $\bbC(\bbS)$ must be followed by an additional operation mapping $N$ dimensions to $M$ dimensions (a fully connected layer---perceptron---, for instance), in which case the number of parameters of $\bbC(\bbS)$ will necessarily depend on $N$ and $M$. Finally, if $\bby_t$ is a single value $\bby$ representing the entire sequence $\{\bbx_t\}_{t=1}^T$, we compute it from the last state alone as $\hby = \rho(\bbC(\bbS) \bbz_{T})$.

Making the hidden state $\bbz_{t}$ a graph signal has several advantages. First, it adds interpretability to the value of this signal with respect to the underlying graph support. For instance, we could analyze the frequency content of the hidden state and compare it with the frequency content of the graph process $\bbx_{t}$. Second, it allows the computation of $\bbz_{t}$ to be done in an entirely local fashion, involving only repeated exchanges with the one-hop neighbors of each node. Making $\bbz_{t}$ a graph signal, however, also implies that we can no longer tune the size of the hidden state which is now fixed at $N$. The size of the hidden state is a fundamental hyperparameter in the design of RNNs since it controls the description capability of the \blue{hidden state}. This can be overcome by introducing graph signal tensors where, instead of a single scalar, a vector of features is assigned to each node.

A graph signal tensor is a function \blue{$\ccalX: \ccalV \to \reals^{F}$} that assigns a vector of dimension $F$ to each node. Each entry of this vector is a \emph{feature}. The signal tensor can be represented as a $N \times F$ matrix $\bbX$, where each column $\bbx^{f} \in \reals^{N}$ is a graph signal corresponding to the values of feature $f$ in all nodes.

The graph convolution operation \eqref{eqn:graphConv} must be extended accordingly, so as to carry out a local, linear transformation mapping the $F$ input features in $\bbX \in \reals^{N \times F}$ to the $G$ output features $\bbY \in \reals^{N \times G}$. This map is implemented by a bank of $FG$ graph filters of order $K$, with filter taps given by $\bba^{fg} = [a_{0}^{fg},\ldots,a_{K-1}^{fg}]$. The graph convolution $\ccalA_{\bbS}: \reals^{N \times F} \to \reals^{N \times G}$ becomes [cf. \eqref{eqn:graphConv}],
\begin{equation} \label{eqn:graphConvExtended}
    \bbY = \ccalA_{\bbS}(\bbX) = \sum_{k=0}^{K-1} \bbS^{k} \bbX \bbA_{k}
\end{equation}
where $\bbA_{k} \in \reals^{F \times G}$ is a matrix satisfying $[\bbA_{k}]_{fg} = a_{k}^{fg}$. We can see that, for the convolution to be local on the graph, the operations that modify $\bbX$ on the left have to respect the sparsity of the graph, while those that modify it on the right can be arbitrary linear operations. The right operations have the role of mixing the features within a single node, using the same linear combination --parameter sharing-- across all nodes.

Given a sequence of graph signal tensors $\{\bbX_{t}\}$, $\bbX_{t} \in \reals^{N \times F}$, we can rewrite equation \eqref{eqn:GRNNhidden} to obtain $H$-feature hidden state tensors $\bbZ_{t} \in \reals^{N \times H}$,
\begin{equation} \label{eqn:GRNNhiddenExt}
    \bbZ_{t} = \sigma \bigg( \ccalA_{\bbS}(\bbX_{t}) + \ccalB_{\bbS}(\bbZ_{t-1}) \bigg)
\end{equation}
where the filter taps are $\bbA_{k} \in \reals^{F \times H}$ and $\bbB_{k} \in \reals^{H \times H}$, $k=0, \ldots, K-1$. Assuming that the target representation is also a graph signal tensor, it can be generalized as $\bbY_{t} \in \reals^{Y \times G}$, which we calculate as
\begin{equation} \label{eqn:GRNNoutputExt}
\bbY_{t} = \rho \bigg( \ccalC_{\bbS}(\bbZ_{t}) \bigg)
\end{equation}
with filter taps $\bbC_{k} \in \reals^{G \times H}$, $k=0, \ldots, K-1$.

By using graph signal tensors to describe the hidden state, we retrieve the ability to tune its descriptive power through the value of $H$. Additionally, note that the output $\bbY_{t}$ can be computed straight from the individual hidden state feature values at each node, minimizing the communication cost. To achieve this, it suffices to make $\ccalC_{S}$ a graph convolution \eqref{eqn:graphConvExtended} with $K=1$, in which case no neighborhood exchanges take place. This architecture only requires node communications at updates of the hidden state $\bbZ_{t}$.


\section{Stability of GRNNs} \label{sec:stability}



The performance of GRNNs (and of graph filters in general) depends on the underlying graph support. If the graph changes, or if it is not estimated accurately, the output of the GRNN can be different than expected. In what follows, we obtain an upper bound on the changes at the output of a GRNN caused by perturbations of the underlying graph. We use this result to quantify how \textit{adaptable} GRNNs are to time-varying scenarios and transfer learning \cite{gama19-neurips}. We focus on single-feature GRNNs \eqref{eqn:GRNNhidden}-\eqref{eqn:GRNNoutput} for simplicity, but results for the multi-feature case carry out similarly \cite{Gama19-Stability}.

Let $\bbS$ be the GSO of a given graph, and let $\tbS$ be the GSO of the graph resulting from a perturbation of this graph. Let us first consider the case of node relabelings, in which $\tbS = \bbP^{\Tr} \bbS \bbP$. The matrix $\bbP$ is a permutation matrix $\bbP \in \ccalP$ with
\begin{equation}
	\ccalP = \big\{ \bbP \in \{0,1\}^{N \times N} : \bbP \bbone = \bbone , \bbP^{\Tr} \bbone = \bbone \big\}.
\end{equation}
If the perturbed graph is simply a permutation of the original graph, then the output of the GRNN running on the permuted graph is the permutation of the output of the GRNN running on the original graph.

%
\begin{proposition} \label{prop:permutationEquivariance}
Let $\bbS$ be a GSO and $\tbS = \bbP^{\Tr} \bbS \bbP$ be a permutation of this GSO, for some permutation matrix $\bbP \in \ccalP$. Let $\bbx_{t}$ be a graph signal and $\tbx_{t} = \bbP^{\Tr} \bbx_{t}$ the permuted version of the signal. Then, it holds that
\begin{align}
	\tbz_{t} &
		= \sigma(\bbA(\tbS) \tbx_{t} + \bbB(\tbS) \tbz_{t-1}) 
		= \bbP^{\Tr} \bbz_{t} \\
	\tby_{t} &
		= \rho (\bbC(\tbS) \tbz_{t}) = \bbP^{\Tr} \bby_{t} \quad \mbox{for all } t.
\end{align}
\end{proposition}
%
\begin{proof}
Refer to Appendix \ref{sec:appendixA}.
\end{proof}

\noindent Proposition~\ref{prop:permutationEquivariance} states that GRNNs are independent of any chosen node labeling. Note that this result holds irrespective of whether we know the value of $\bbP$ or not. It also indicates that GRNNs are able to exploit the internal symmetries of graph processes in the course of learning. This means that by learning how to process a signal on a given part of the graph, GRNNs are also learning to process it in all other parts of the graph that are topologically symmetric. The permutation equivariance property can thus be seen as an implicit mechanism of \emph{data augmentation}.

When considering more general perturbations $\tbS \in\reals^{N \times N}$, Proposition~\ref{prop:permutationEquivariance} suggests that distances should be measured modulo permutations. In order to do so, we introduce the notion of relative perturbation in Definition \ref{def:errorSet}.
%
\begin{definition}[Relative perturbation matrices] \label{def:errorSet}
Given GSOs $\bbS$ and $\tbS$, we define the set of relative perturbation matrices modulo permutation as
\begin{equation} \label{eqn:errorSet}
	\ccalE(\bbS,  \tbS) = \left\{ \bbE \! \in \! \reals^{N \times N} : \bbP^{\Tr} \tbS \bbP = \bbS + \bbE \bbS + \bbS \bbE^{\Tr} , \bbP \in \ccalP \right\}.
\end{equation}
\end{definition}
\noindent We define the distance between two graphs described by $\bbS$ and $\tbS$ respectively as
\begin{equation} \label{eqn:graphDistance}
	d(\bbS, \tbS) = \min_{\bbE \in \ccalE(\bbS, \hbS)} \| \bbE \|.
\end{equation}
Notice that if $\tbS$ is a permutation of $\bbS$, then $d(\bbS, \tbS) = 0$.

To understand the effect of graph perturbations on the output of GRNNs, we first look at their effect on graph convolutions \eqref{eqn:graphConv}. In particular, we leverage the graph Fourier transform (GFT) \cite{Sandryhaila14-DSPGfreq} to analyze the spectral representations of the convolved graph signals. Let $\bbS = \bbV \bbLambda \bbV^{\Hr}$ be the eigendecomposition of the GSO $\bbS$, where $\bbV = [\bbv_{1},\ldots,\bbv_{N}]$ is the orthogonal matrix of eigenvectors, and $\bbLambda = \diag(\lambda_{1},\ldots,\lambda_{N})$ is the diagonal matrix of eigenvalues $\lambda_{n}$. The GFT of a signal is computed by projecting the signal on the graph's eigenvector basis. The GFT of a convolved graph signal is thus
\begin{equation} \label{eqn:graphConvGFT}
	\bbV^{\Hr} \bbA(\bbS) \bbx = \bbV^{\Hr} \sum_{k=0}^{K-1} a_{k} \bbV \bbLambda^{k} \bbV^{\Hr} \bbx = \bbA(\bbLambda) \left( \bbV^{\Hr} \bbx \right)
\end{equation}
where $\bbA(\bbLambda)$ is a diagonal matrix such that $[\bbA(\bbLambda)]_{n} = \sum_{k=0}^{K-1} a_{k} \lambda_{n}^{k} = a(\lambda_{n})$. We refer to the function $a(\lambda)$ as the \emph{frequency response} of the filter, given by
\begin{equation} \label{eqn:freqResponse}
	a(\lambda) = \sum_{k=0}^{K-1} a_{k} \lambda^{k}.
\end{equation}
For a given graph, this frequency response gets instantiated on the graph's eigenvalues $\{\lambda_{n}\}$ as $\bbA(\bbLambda)$, determining the specific effect that the filter has on the input due to the underlying support \eqref{eqn:graphConvGFT}. Note, however, that the general expression of the frequency response \eqref{eqn:freqResponse} only depends on the filter taps $\{a_{k}\}$, which are independent of the graph.

The results here derived are for graph filters with \textit{integral Lipschitz} frequency response.
%
\begin{definition}[Integral Lipschitz filters] \label{def:integralLipschitz}
	Given a set of filter taps $\{a_{k}\}$, we say that the filter $\bbA(\bbS)$ [cf. \eqref{eqn:graphConv}] is integral Lipschitz if \blue{there exists $C$ such that} its frequency response $a(\lambda)$ [cf. \eqref{eqn:freqResponse}] satisfies
	\begin{equation} \label{eqn:integralLipschitz}
		|a(\lambda_{2}) - a(\lambda_{1}) | \leq C \frac{|\lambda_{2}-\lambda_{1}|}{|\lambda_{1}+\lambda_{2}|/2}
	\end{equation}
	for all $\lambda_{1},\lambda_{2} \in \reals$.
\end{definition}
%
\noindent Integral Lipschitz filters also satisfy $|\lambda a'(\lambda)| \leq C$, where $a'(\lambda)$ is the derivative of $a(\lambda)$. This condition is reminiscent of the scale invariance of wavelet transforms \cite[Chapter 7]{Daubechies92-Wavelets}.

\blue{Under the following assumptions, we prove that GRNNs built from integral Lipschitz filters are stable to relative perturbations in Theorem \ref{thm:stability}}.

\blue{
\begin{assumption} \label{as1}
The filters $\bbA$, $\bbB$ and $\bbC$ of the GRNN \eqref{eqn:GRNNhidden}-\eqref{eqn:GRNNoutput} are integral Lipschitz [cf. \eqref{eqn:integralLipschitz}] with constants $C_{\bbA}$, $C_{\bbB}$ and $C_{\bbC}$ and normalized filter height $\|\bbA\|=\|\bbB\|=\|\bbC\|=1$, respectively.
\end{assumption}

\begin{assumption}\label{as2}
The pointwise nonlinearities $\sigma$ and $\rho$ \eqref{eqn:GRNNhidden}-\eqref{eqn:GRNNoutput} are normalized Lipschitz, i.e. $|\sigma(b) - \sigma(a)| \leq |b-a|$ for all $a,b \in \reals$, and satisfy $\sigma(0)=\rho(0)=0$.
\end{assumption}

\begin{assumption} \label{as3}
The initial hidden state is identically zero, i.e. $\bbz_0 = \boldsymbol{0}$.
\end{assumption}

\begin{assumption} \label{as4}
The inputs $\bbx_t$ satisfy $\|\bbx_t\|\leq\|\bbx\|=1$ for every $t$.
\end{assumption}
}

%
\begin{theorem}[Stability of GRNNs] \label{thm:stability}
	\blue{Consider two graphs with $N$ nodes represented by the GSOs $\bbS = \bbV \bbLambda \bbV^{\Hr}$ and $\tbS$.}
	 Let $\bbE = \bbU \bbM \bbU^{\Hr} \in \ccalE(\bbS, \tbS)$ be a relative perturbation matrix [cf. \eqref{eqn:errorSet}] such that [cf. \eqref{eqn:graphDistance}]
	\begin{equation} \label{eqn:distanceCondition}
		d(\bbS, \tbS) \leq \| \bbE \| \leq \varepsilon.
	\end{equation}
	\blue{Let $\bby_t$ and $\tby_t$ be the outputs of GRNNs \eqref{eqn:GRNNhidden}-\eqref{eqn:GRNNoutput} running on $\bbS$ and $\tbS$ respectively, and satisfying AS\ref{as1} through AS\ref{as4}.}
	 Then, it holds that
	\begin{equation} \label{eqn:stability}
		\min_{\bbP \in \ccalP} \| \bby_{t} - \bbP^{\Tr} \tby_{t} \| \leq C(1+\sqrt{N} \delta)(t^2+3t) \varepsilon\ + \ccalO(\varepsilon^{2})
	\end{equation}
	\blue{where $C$ is the maximum filter constant,
	 \[C = \max\{C_{\bbA},C_{\bbB},C_{\bbC}\}\]}
	 and $\delta = (\|\bbU-\bbV\| + 1)^{2} - 1$ measures the eigenvector misalignment between the GSO $\bbS$ and the error matrix $\bbE$.
\end{theorem}
%
\begin{proof}
Refer to Appendix \ref{sec:appendixB}.
\end{proof}

Theorem \ref{thm:stability} states that, for a graph process of length $t=T$, the output of a GRNN is Lipschitz stable to relative graph perturbations [cf. Def.~\ref{def:errorSet}] with constant $C(1+\sqrt{N} \delta)(T^2+3T)$. We see that the stability of a GRNN depends on the Lipschitz filter constant $C$. While this is a \textit{design} parameter that could be set at a fixed value, it is usually learned from data through the filter taps of $\bbA$, $\bbB$ and $\bbC$. The term $(1 + \delta\sqrt{N})$ measures the eigenvector misalignment and is a property of the graph perturbation. Unlike $C$, it cannot be controlled by design. Finally, the stability of GRNNs depends polynomially on the length $T$ of the process, with $T^2 + 3T$. The linear term arises from sequential applications of the filters $\bbA$, $\bbB$ and $\bbC$, and the square term is a result of the recurrence on $\bbz_t$. We note that $T$ can be controlled by restraining the length of the graph processes that we consider, or by splitting them in multiple shorter processes.



\section{Gated GRNN architectures} \label{sec:gatedGRNN}



One problem that can arise from long sequences is that of \emph{vanishing (exploding) gradients}.
Traditional RNN architectures suffer from this problem when the input sequence contains long term dependencies \cite{pascanu2013difficulty, bengio1994learning}. The same holds for GRNNs when the eigenvalues of $\bbB(\bbS)$ are smaller (or larger) than $1$. 
RNN architectures typically address problems associated with long term dependencies by the addition of time gating mechanisms \cite[Chapter 10]{Goodfellow16-DeepLearning}, which can be naturally extended to GRNNs (Sec.~\ref{subsec:timeGating}). 

When dealing with graph processes, we may also encounter what we call the problem of \textit{vanishing gradients in space} (in contrast with the aforementioned problem of \textit{vanishing gradients in time}). Even if the eigenvalues of $\bbB(\bbS)$ are well-behaved, some nodes or paths of the graph might get assigned more importance than others in long range exchanges, leading to \emph{spatial imbalances} that make it challenging to encode certain graph spatial dependencies. This problem can be explained by the fact that the matrix multiplications by $\bbB(\bbS)$ are actually multiplications by powers of $\bbS$. As an example, consider a graph where some components have higher connectivities than others. For large $t$, the matrix entries associated with nodes belonging to highly connected components will get densely populated, overshadowing other local, sparser structures of these components and making it harder to distinguish long range processes that are local on the graph. \blue{While in GNNs long range graph dependencies have been addressed by the use of Lanczos filtering methods \cite{susnjara2015accelerated,liao2019lanczosnet}, or by computing the spectral response of graph convolutions in Krylov form \cite{luan2019break}, these techniques require calculating a basis that depends on both the graph and the signal, and, as such, are costly to translate to GRNNs dealing with signals that change over time.}

To attenuate these issues, we propose to add a more comprehensive \textit{gating mechanism} to GRNNs. Similarly to the gates employed in traditional RNNs, the gates that we consider are operators acting on the current input and previous state to control how much of the input should be taken into account and how much past information should be \emph{remembered} (or \emph{forgotten}) in the computation of the new state. These gating operators are updated at every step of the sequence and, as such, they are able to create multiple dependency paths between states and inputs in both time and space. This allows for both short and long term dependencies to be encoded by the model without getting assigned exponentially smaller or larger weights. Adding gating to GRNNs yields the Gated GRNN (GGRNN), in which $\bbZ_t$ is computed as
\begin{equation} \label{eqn:gatingGeneric}
    \bbZ_{t} = \sigma \bigg( \hcalQ \left\{ \ccalA_{\bbS}(\bbX_{t}) \right\} + \kcalQ \left\{ \ccalB_{\bbS} (\bbZ_{t-1} ) \right\} \bigg)
\end{equation}
and where $\hcalQ: \reals^{N \times H} \to \reals^{N \times H}$ stands for the \textit{input gate} operator and $\kcalQ: \reals^{N \times H} \to \reals^{N \times H}$ for the \textit{forget gate} operator.

Depending on the choice of gating strategy, which we will discuss in the following subsections, $\hat{\ccalQ}$ and $\check{\ccalQ}$ take on different forms. What they all have in common is that their parameters are themselves calculated as the output of GRNNs. The GRNN used to calculate the input gate has \textit{input gate state} $\hbZ_t \in \reals^{N \times \hhatH}$ given by 
\begin{equation} \label{eqn:inputgatestate}
\hbZ_{t} = \hhatsigma \bigg( \hcalA_{\bbS} (\bbX_{t})  + \hcalB_{\bbS} (\hbZ_{t-1}) \bigg)
\end{equation}
and the GRNN used to calculate the forget gate has \textit{forget gate state} $\check{\bbZ}_t \in \reals^{N \times \chkH}$,
\begin{equation} \label{eqn:forgetgatestate}
\cbZ_{t} = \check{\sigma} \bigg( \kcalA_{\bbS} (\bbX_{t})  + \kcalB_{\bbS} (\cbZ_{t-1}) \bigg)
\end{equation}
where $\hcalA_{\bbS}$, $\hcalB_{\bbS}$, $\kcalA_{\bbS}$ and $\kcalB_{\bbS}$ are graph convolutions [cf. \eqref{eqn:graphConvExtended}] with filter taps $\hbA_{k} \in \reals^{F \times \hhatH}$, $\hbB_{k} \in \reals^{\hhatH \times \hhatH}$, $\cbA_{k} \in \reals^{F \times \chkH}$ and $\cbB_{k} \in \reals^{\chkH \times \chkH}$.

To tackle the different time and spatial imbalance scenarios described in this section, we envision three gating strategies: time (Sec.~\ref{subsec:timeGating}), node (Sec.~\ref{subsec:nodeGating}) and edge gating (Sec.~\ref{subsec:edgeGating}).

\subsection{Time gating} \label{subsec:timeGating}

In the \textit{Time Gated GRNN} (t-GGRNN), the input and forget gate operators $\hcalQ$ and $\kcalQ$ take form
\begin{equation} \label{eqn:tGGRNN}
\begin{aligned}
    \hcalQ \left\{ \ccalA_{\bbS}(\bbX_{t}) \right\} & = \hhatq_{t}  \ccalA_{\bbS}(\bbX_{t}) \\
    \kcalQ \left\{ \ccalB_{\bbS}(\bbZ_{t}) \right\} & = \chkq_{t}  \ccalB_{\bbS}(\bbZ_{t})
\end{aligned}
\end{equation}
with $\hhatq_{t} \in [0,1]$ and $\chkq_{t} \in [0,1]$ computed as
\begin{equation} \label{eqn:tgates_comp}
\begin{aligned}
    \hhatq_{t} = \mathrm{sigmoid} (\hbc^{\Tr} \mathrm{vec}(\hbZ_{t})) \\
    \chkq_{t} = \mathrm{sigmoid} (\cbc^{\Tr} \mathrm{vec}(\check{\bbZ}_{t})) \\
\end{aligned}
\end{equation}
and where $\hbc \in \reals^{\hhatH N}$ and $\cbc \in \reals^{\chkH N}$ are learnable parameters. 

Time gating addresses the problem of vanishing gradients in time by learning scalar gates between $0$ and $1$ and multiplying the input and state variables by these gates, thus compensating for imbalanced gradient paths associated with eigenvalues that are too small or too large. We refer to this strategy as time gating because it only acts on time dependencies, \emph{shutting down} the whole input and/or the whole previous state at each time instant as needed, without discriminating between nodes. Here, note that the number of parameters necessary to map the state to the input and forget gates are dependent on the size of the graph, because all of the graph signal components must be mapped onto scalar variables. 

The basic architecture of a t-GGRNN resembles that of the Long Short-Term Memory units (LSTMs) used to process regular data sequences \cite[Chapter~10]{Goodfellow16-DeepLearning}, with the difference that LSTMs have an output gate in addition to the input and forget gates. The input and forget gates of a LSTM are calculated in the same way $\hhatq_t$ and $\chkq_t$ in \eqref{eqn:tgates_comp} would be if we considered the directed cycle graph.
Another common gated architecture for regular data are Gated Recurrent Units (GRUs)\cite[Chapter~10]{Goodfellow16-DeepLearning}, which are even simpler than LSTMs where only one gating variable $u_t \in [0,1]$ acts as the forget gate of LSTMs, and where the input gate is replaced by $1 - u_t$. The GRU architecture can be readily extended to t-GGRNNs.

\subsection{Node gating} \label{subsec:nodeGating}

In some cases, having the input and forget gates of a gated GRNN be scalars is limiting because the short/long term time interactions of the graph process might vary across nodes. This is especially true of graph processes that are, in reality, some unknown composition of processes happening independently at each node and/or on multiple, possibly non-disjoint, subgraphs of the original graph. In the Node Gated GRNN (n-GGRNN), we address this by defining the input gate and forget gate operators $\hcalQ$ and $\kcalQ$ as
\begin{equation}
\begin{aligned}
\hcalQ \left\{ \ccalA_{\bbS}(\bbX_{t}) \right\} & = \diag(\hbq_{t})  \ccalA_{\bbS}(\bbX_{t}) \\
\kcalQ \left\{ \ccalB_{\bbS}(\bbZ_{t}) \right\} & = \diag(\cbq_{t})  \ccalB_{\bbS}(\bbZ_{t})
\end{aligned}
\end{equation}
with parameters $\hbq_t \in [0,1]^N$ and $\cbq_t \in [0,1]^N$ given by
\begin{equation} \label{eqn:ngates_comp}
\begin{aligned}
\hbq_{t} = \mathrm{sigmoid} \left( \hat{\ccalC}_{\bbS}(\hbZ_{t}) \right) \\
\cbq_{t} = \mathrm{sigmoid} \left( \check{\ccalC}_{\bbS}(\check{\bbZ}_{t}) \right) \\
\end{aligned}
\end{equation}
and where, now, the learnable parameters are the filter taps of the graph convolutions $\hat{\ccalC}_{\bbS}$ and $\check{\ccalC}_{\bbS}$, given by $\hbC_{k} \in \reals^{1 \times \hhatH}$ and $\check{\bbC}_{k} \in \reals^{1 \times \chkH}$.

In the n-GGRNN, the gates $\hbq_t$ and $\cbq_t$ are reshaped as the diagonal matrices $\diag(\hbq_{t})$ and $\diag(\cbq_{t})$, which then multiply the input and state variables. The multiplication by $\diag(\hbq_{t})$ and $\diag(\cbq_{t})$ has the role of applying a separate scalar input and forget gate (both taking values between $0$ and $1$) to each node. This allows addressing the problem of vanishing gradients in space by controlling the importance of the input and of the state at the node level and partially shutting down nodes whose signal components can effectively behave as noise in the exchanges involved in some learning tasks. Besides adding flexibility to the gated architecture, n-GGRNNs have the advantage that their number of parameters is independent of the size of the graph, which could not be said about the t-GGRNNs from the previous subsection.

An interesting observation is that the composition of node gating with a graph convolution can be interpreted as the application of a node-varying graph filter \cite{Gama18-NodeVariant}, which, instead of weighing powers of $\bbS$ by scalars as in the LSI-GF [cf. \eqref{eqn:graphConv}], multiplies them by diagonal matrices assigning a different weight to each node. From an implementation standpoint, this is important because it allows simplifying the operations involved in the n-GGRNN.

\subsection{Edge gating} \label{subsec:edgeGating}

In node gating, we control long range graph dependencies by assigning a gate to each node \textit{after} local exchanges have occurred. In edge gating, the gates act within these local exchanges, controlling the amount of information that is transmitted across edges of the graph. The input and forget gate operators take form
\begin{equation}
\begin{aligned}
\hcalQ \left\{ \ccalA_{\bbS}(\bbX_{t}) \right\} & = \ccalA_{\bbS \odot \hbQ_{t}}(\bbX_{t}) \\
\kcalQ \left\{ \ccalB_{\bbS}(\bbZ_{t}) \right\} & = \ccalB_{\bbS  \odot \cbQ_{t}}(\bbZ_{t})
\end{aligned}
\end{equation}
where the shift operators that parametrize the input-to-state and state-to-state convolutions are now $\bbS \odot \hbQ_{t}$ and $\bbS \odot \cbQ_{t}$ respectively, with $\hbQ_{t}, \cbQ_{t} \in [0,1]^{N \times N}$.
$\hbQ_t$ and $\cbQ_t$ are calculated as
\begin{equation} \label{eqn:edgeGateComps}
\begin{aligned}[]
[\hbQ_{t}]_{ij} = \mathrm{sigmoid} \left( \hbc^{\Tr} [\bbdelta_{i}^{\Tr}\hbZ_{t}\hbC ||  \bbdelta_{j}^{\Tr}\hbZ_{t}\hbC]^{\Tr}  \right) \\
[\cbQ_{t}]_{ij} = \mathrm{sigmoid} \left( \cbc^{\Tr} [\bbdelta_{i}^{\Tr}\cbZ_{t}\cbC ||  \bbdelta_{j}^{\Tr}\cbZ_{t}\cbC]^{\Tr} \right)
\end{aligned}
\end{equation}
where $\bbdelta_i$ stands for the one-hot column vector with $[\bbdelta_i]_i=1$ and $||$ is the horizontal concatenation operation. The learnable parameters are $\hbC \in \reals^{\hhatH \times \hhatH'}$, $\hbc \in \reals^{2\hhatH' \times 1}$, $\cbC \in \reals^{\chkH \times \chkH'}$ and $\cbc \in \reals^{2\chkH' \times 1}$, and $\hhatH'$ and $\chkH'$ are arbitrary numbers of \emph{intermediate} features. 

Effectively, $\hbQ_t$ and $\cbQ_t$ scale the weight of each edge by a value between $0$ and $1$. When this value is $0$, the edge exchange is completely shut off, which can be helpful in GRNNs running on graphs with noisy or spurious edges, e.g. graphs built from sample covariance matrices. Note that each edge input gate $[\hbQ_{t}]_{ij}$ and forget gate $[\cbQ_{t}]_{ij}$ is computed individually, avoiding unnecessary computations for pairs $(i,j)$ that do not correspond to edges of the graph.

In practice, the computations carried out in equation \eqref{eqn:edgeGateComps} are implemented as Graph Attention Networks (GANs) \cite{Velickovic18-GraphAttentionNetworks}, whose \textit{attention coefficients} play the role of $[\hbQ_{t}]_{ij}$ and $[\cbQ_{t}]_{ij}$. Specifically tailored to graphs, GATs are attention mechanisms that generate meaningful representations of graph signals by incorporating the importance of a node's features to its neighbors in the extraction of subsequent features. This importance is learned in the form of attention coefficients between nodes $i$ and $j$ that are connected by an edge, and is calculated by applying a linear transformation and a nonlinearity to the their concatenated features.
Following normalization (either by a nonlinearity such as the sigmoid or by some other normalizing operation), the attention coefficients of GATs taking in $\hbZ_t$ and $\check{\bbZ}_t$ are well-suited implementations of the input and forget edge gates $[\hbQ_t]_{ij}$ and $[\cbQ_t]_{ij}$.

Similarly to how the composition of node gating with a graph convolution could be interpreted as a node-varying graph filter, composing edge gating with LSI-GFs can be seen as a particular implementation of an edge-varying graph filter \cite{contino2017distributed, isufi19-edgenets, Isufi20-EdgeNets}. Edge-varying graph filters are such that each edge is parametrized independently in multiplications by the GSO, which is precisely what happens when edge gates are applied to $\bbS$ in the input-to-state and state-to-state convolutions.

\subsection{\blue{Stability}} \label{subsec:gatingStability}

\blue{Since the parameters of the gate operators $\hcalQ$ and $\kcalQ$ are themselves the outputs of GRNNs, it is natural to ask whether the stability result from Theorem \ref{thm:stability} carries over to the gated GRNN in \eqref{eqn:gatingGeneric}. In order to analyze its stability without focusing on a specific type of gating, we will describe the
parameters of $\hcalQ$ and $\kcalQ$ as $\hat{\bbtheta}$ and $\check{\bbtheta}$, and the architectures used to predict these parameters from the state variables $\hbz_t$ \eqref{eqn:inputgatestate} and $\cbz_t$ \eqref{eqn:forgetgatestate} as generic models $\hat{\bbPhi}_{\bbS}$ and $\check{\bbPhi}_\bbS$ (the subscript $\bbS$ indicates that the graph is a hyperparameter). Strictly speaking, the parameters of the input and forget gates are obtained as $\hat{\bbtheta} = \hat{\bbPhi}_\bbS(\hbz_t)$ and $\check{\bbtheta} = \check{\bbPhi}_\bbS(\cbz_t)$, and we can write $\hcalQ = \hcalQ_{\hat{\bbtheta}}$ and $\kcalQ = \kcalQ_{\check{\bbtheta}}$. Note that the generic architectures $\hat{\bbPhi}_\bbS$ and $\check{\bbPhi}_\bbS$ can always be particularized to different gating mechanisms. In the case of time gating, $\hat{\bbPhi}_\bbS$ and $\check{\bbPhi}_\bbS$ are fully connected layers and $\hat{\bbtheta}$, $\check{\bbtheta}$ are scalars [cf. \eqref{eqn:tgates_comp}]; in node gating, $\hat{\bbPhi}_\bbS$ and $\check{\bbPhi}_\bbS$ are GNNs and $\hat{\bbtheta}$, $\check{\bbtheta}$ are vectors [cf. \eqref{eqn:ngates_comp}]; and, in edge gating, $\hat{\bbPhi}_\bbS$ and $\check{\bbPhi}_\bbS$ are the attention mechanism of a GAN and $\hat{\bbtheta}$, $\check{\bbtheta}$ are matrices [cf. \eqref{eqn:edgeGateComps}]. 

To prove stability of gated GRNNs, we also need the following assumptions on the gate operators $\hcalQ$, $\kcalQ$ and on the parameter learning models ${\hat{\bbPhi}_\bbS}$, $\check{\bbPhi}_\bbS$.

\begin{assumption} \label{as5}
In the induced operator norm, $\hcalQ=\hcalQ_{\hat{\bbtheta}}$ and $\kcalQ=\kcalQ_{\check{\bbtheta}}$ are $Q$-Lipschitz with respect to $\hat{\bbtheta}$ and $\check{\bbtheta}$, i.e. 
\[\|\hcalQ_{\hat{\bbtheta}^{(1)}}-\hcalQ_{\hat{\bbtheta}^{(2)}}\| \leq Q\|\hat{\bbtheta}^{(1)}-\hat{\bbtheta}^{(2)}\|.\]
\end{assumption}
 
\begin{assumption} \label{as6}
The models ${\hat{\bbPhi}_\bbS}$ and $\check{\bbPhi}_\bbS$ are $\phi_1$-Lipschitz functions of $\hbz_t$ and $\cbz_t$, i.e. $\|{\hat{\bbPhi}}_\bbS(\hbz^{(1)}_t) - {\hat{\bbPhi}}_\bbS(\hbz^{(2)}_t)\| \leq \phi_1 \|\hbz^{(1)}_t-\hbz^{(2)}_t\|$.
\end{assumption}

\begin{assumption} \label{as7}
In the induced operator norm, ${\hat{\bbPhi}_\bbS}$ and $\check{\bbPhi}_\bbS$ are $\phi_2$-Lipschitz with respect to the graph, i.e. $\|{\hat{\bbPhi}}_{\bbS^{(1)}} - {\hat{\bbPhi}}_{\bbS^{(2)}}\| \leq \phi_2 \|\bbS^{(1)}-\bbS^{(2)}\|$.
\end{assumption}

AS\ref{as5} is important because it allows expressing the norm of the difference operator in terms of the norm difference of the operators' parameters. It is also worth noting that assumptions AS\ref{as6} and AS\ref{as7} are not too restrictive, as we discuss in more detail in Remark \ref{rmk:restrictive}.  

\begin{theorem}[Stability of Gated GRNNs] \label{thm:GGRNNstability}
	Consider two graphs with $N$ nodes represented by the GSOs $\bbS = \bbV \bbLambda \bbV^{\Hr}$ and $\tbS$.
	 Let $\bbE = \bbU \bbM \bbU^{\Hr} \in \ccalE(\bbS, \tbS)$ be a relative perturbation matrix [cf. \eqref{eqn:errorSet}] such that [cf. \eqref{eqn:graphDistance}]
	\begin{equation} \label{eqn:distanceCondition2}
		d(\bbS, \tbS) \leq \| \bbE \| \leq \varepsilon.
	\end{equation}
	Let $\bby_t$ and $\tby_t$ be the outputs of gated GRNNs [cf. \eqref{eqn:gatingGeneric} and \eqref{eqn:GRNNoutput}] with $F=H=1$ feature running on $\bbS$ and $\tbS$ respectively, and satisfying AS\ref{as1} through AS\ref{as7}.
	 Then, it holds that
	\begin{align} \label{eqn:ggrnn_stability}
	\begin{split}
		\min_{\bbP \in \ccalP} \| \bby_{t} - \bbP^{\Tr} \tby_{t} \| &\leq C(1+\delta\sqrt{N})(3t+t^2)\varepsilon \\
&+ Q\left(\phi_2+\phi_1C(1+\delta\sqrt{N})\right)t^3\varepsilon\\
&+Q\phi_1C(1+\delta\sqrt{N})t^4\varepsilon+\ccalO(\varepsilon^2)
\end{split}
	\end{align}
	where $C$ is the maximum filter constant, 
	\[C =\max\{C_{\bbA},C_{\bbB},C_{\hbA},C_{\hbB},C_{\cbA},C_{\cbB},C_{\bbC}\}\]
    and $\delta = (\|\bbU-\bbV\| + 1)^{2} - 1$ measures the eigenvector misalignment between the GSO $\bbS$ and the error matrix $\bbE$.

\end{theorem}

We observe that the stability constant of gated GRNNs is equal to the stability constant of the non-gated GRNN [cf. Theorem \ref{thm:stability}] plus a term that depends on the third and fourth powers of $t$. This is because the parameters of the input and forget gates are generated by GRNNs, which add two more recurrence relationships to the architecture. 
Regardless of the gating mechanism (time, node or edge gates), this additional term can be adjusted by tuning $Q$ and, especially, $\phi_1$ and $\phi_2$, which are specific to the architectures $\hat{\bbPhi}$ and $\check{\bbPhi}$. 

\begin{remark} \label {rmk:restrictive}
Fully connected and convolutional layers with most conventional activation functions are Lipschitz stable to input perturbations \cite{virmaux2018lipschitz}, so AS\ref{as6} is generally satisfied for all types of gating (in edge gating, the most common attention mechanism is a multi-layer perceptron \cite{Velickovic18-GraphAttentionNetworks}). In time gating, the fully connected layers that make up $\hat{\bbPhi}$ and $\check{\bbPhi}$ do not depend on $\bbS$, so AS\ref{as7} is also satisfied automatically. Since in node gating $\hat{\bbPhi}$ and $\check{\bbPhi}$ are GNNs, AS\ref{as7} is guaranteed by \cite[Theorem 4]{Gama19-Stability} as long as the graph convolutions are integral Lipschitz [cf. Definition \ref{def:integralLipschitz}] and the activation functions are normalized Lipschitz [cf. Assumption \ref{as2}]. 
Finally, we can expect AS\ref{as7} to hold for edge gating at least in cases where the edges of the graph are preserved, since the attention coefficients only depend on there being an edge between two nodes, but not on the edge weight. Also note that edge preservation is enforced by the definition of the minimum relative perturbation matrix, which measures how close two graphs are to being permutations of one another [cf. \eqref{eqn:graphDistance}].
\end{remark}

}


\section{Numerical Experiments} \label{sec:sims}



%

\begin{figure*}[t]
	\centering
	\begin{subfigure}{.28\textwidth}
		\centering
		\includegraphics[width=\textwidth]{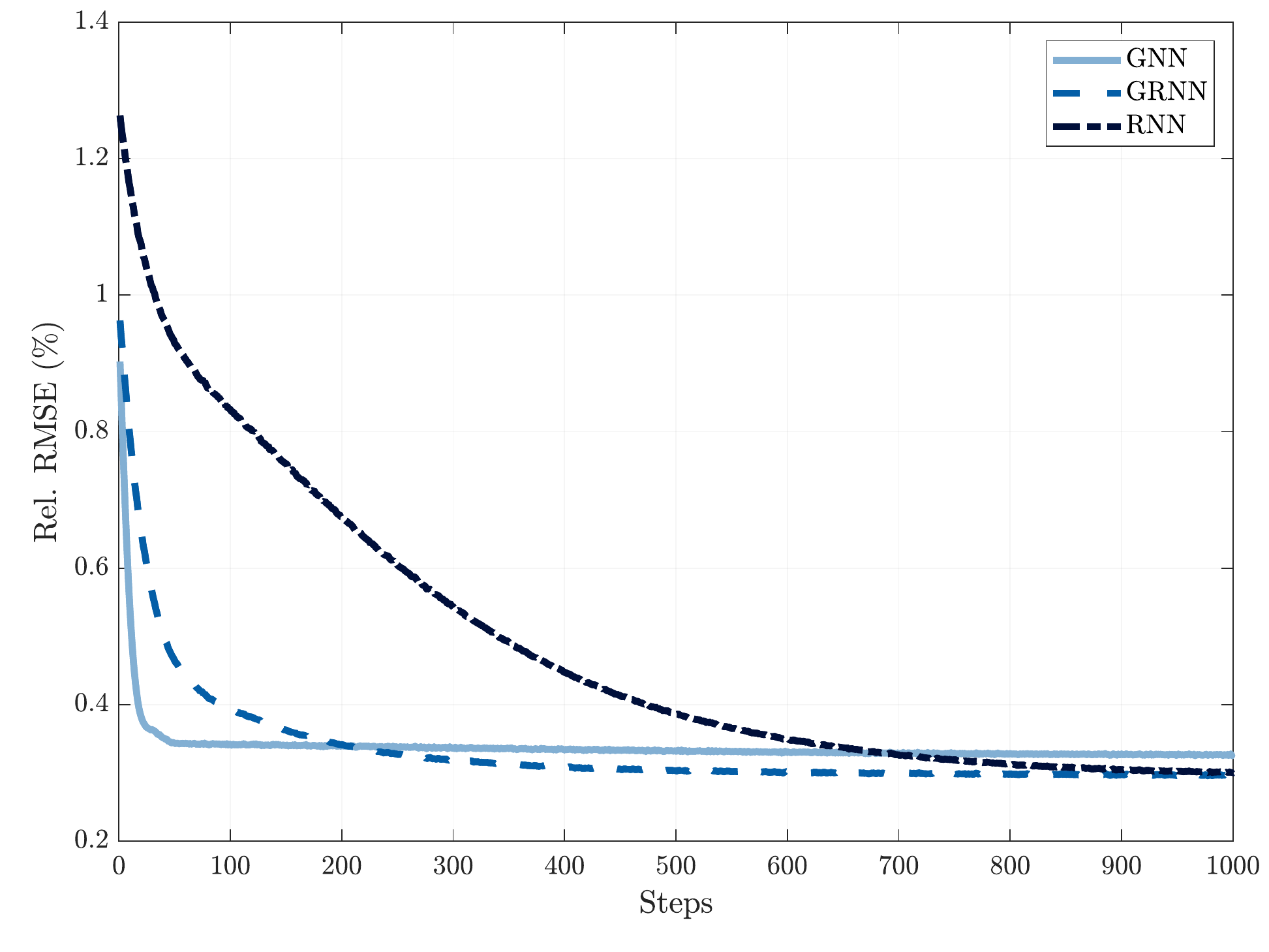}
		\caption{}
		\label{evolution}
	\end{subfigure}
	\hspace{2em}
	\begin{subfigure}{.28\textwidth}
		\centering
		\includegraphics[width=\textwidth]{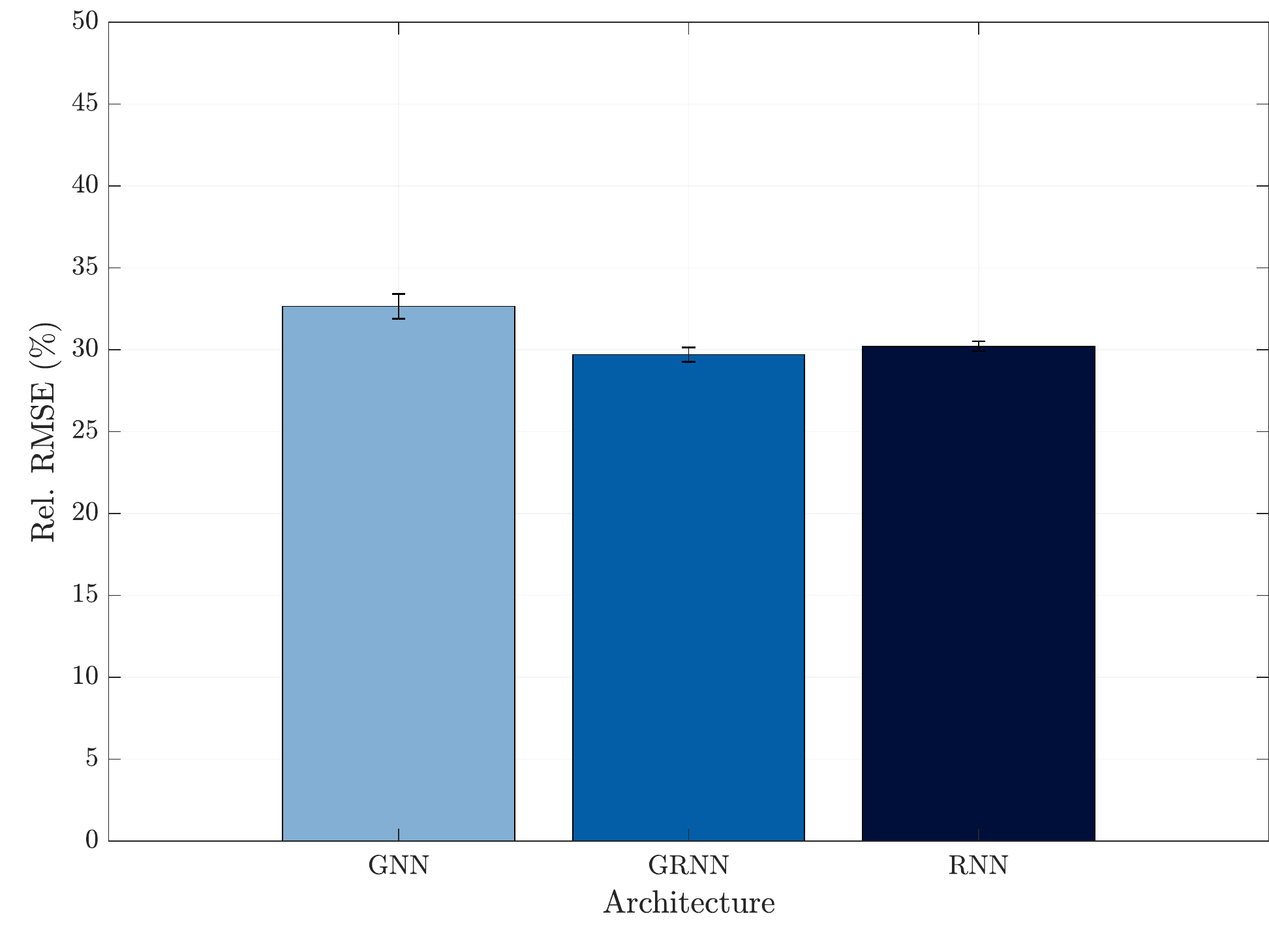} 
		\caption{}
		\label{10000}		
	\end{subfigure}
		\hspace{2em}
	\begin{subfigure}{.28\textwidth}
		\centering
		\includegraphics[width=\textwidth]{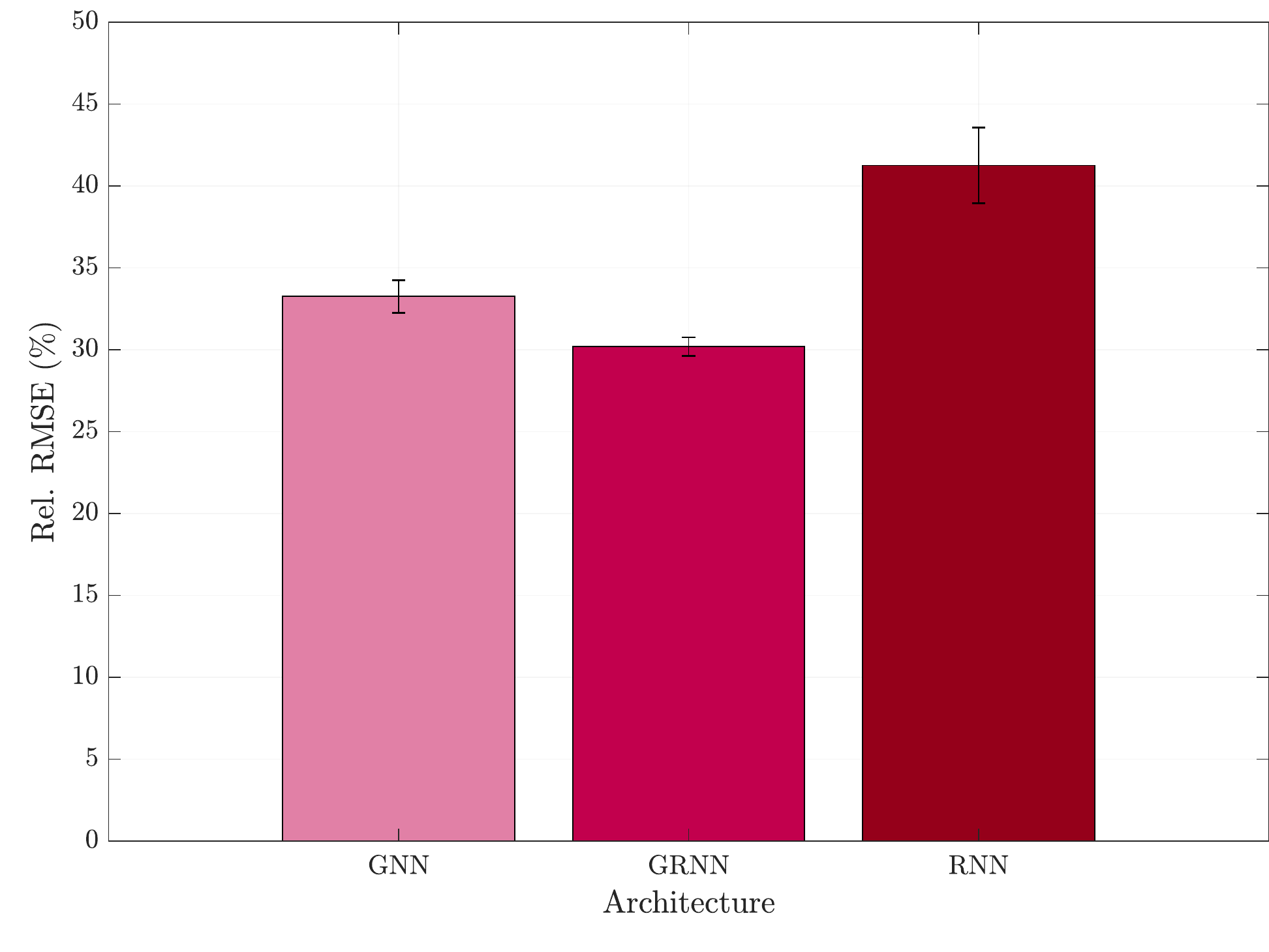} 
		\caption{}
		\label{5000}		
	\end{subfigure}
	\caption{\blue{$5$-step prediction using a GNN, a GRNN and a RNN. \subref{evolution} Training rRMSE evolution over 1000 training steps. \subref{10000} Average test rRMSEs for 5 graphs and 5 data realizations, using a 10,000-sample training dataset. \subref{5000} Average test rRMSEs for 5 graphs and 5 data realizations, using a 5,000-sample training dataset.}}
	\label{fig:grnn_basic}
\end{figure*}	

In this section, we present a series of numerical experiments where the GRNN is compared with GNNs and RNNs, and where the advantages of time, node and edge gating are analyzed. In the first four experiments (subsections \ref{sbs:comparison_gnn_rnn}, \ref{sbs:ar}, \ref{sbs:diffusion}, \ref{sbs:covariance}), we use synthetic data to simulate the problem of $k$-step prediction, where, given instantaneous observations of a synthetic graph process, the goal is to predict the graph signals observed $k$ steps ahead. In each subsection, a different type of process is considered to assess the advantages of various gating strategies in different scenarios. 
The fifth experiment (subsection \ref{sbs:earthquake}) uses earthquake data from New Zealand's Geonet database \cite{geonet} to predict the region of origin of each earthquake registered between June 17, 2019 and July 17, 2019. The data consists of seismograph readings from a network of $N=59$ seismographs immediately before each earthquake, and all earthquakes are assigned a class corresponding to one out of $C=11$ regions. \blue{In subsection \ref{sbs:traffic}, we use real traffic data measured by $N=207$ speed sensors to perform traffic forecasting in the Los Angeles metropolitan area. This data was collected between March and June 2012 and is aggregated in the METR-LA dataset \cite{jagadish2014big}, which is commonly used for benchmark in traffic forecasting. In the last experiment (subsection \ref{sbs:epidemic}), we use real 2013 data from a high school in Marseilles to build a friendship network with $N=134$ nodes where we simulate the spread of an infectious disease using the SIR (Susceptible-Infected-Recovered) model. Different GRNN models are then trained to solve a binary node classification problem aiming to predict which nodes of the network will be infected in 8 days.
}

\blue{In the experiments of subsections \ref{sbs:earthquake} through \ref{sbs:epidemic}, we also include comparisons with other gated graph recurrent architectures from the literature, namely the DCRNN \cite{li2017diffusion} and the GCRN \cite{seo2018structured}. Both of them use a node gating mechanism and slight variations of their architectures can be obtained by particularizing the GSO in equation \eqref{eqn:gatingGeneric} to $\bbS = \diag(\bbA\boldsymbol{1})\bbA + \diag(\bbA^\Tr\boldsymbol{1})\bbA^\Tr$ (the random walk matrix) and $\bbS = \bbI - \bbD^{-1/2}\bbA\bbD^{-1/2}$ (the normalized Laplacian) respectively. Note that, because the GSO of the GCRN is the normalized Laplacian, this architecture can only be applied to problems where the graph is undirected, which is why we do not include it in the earthquake epicenter estimation experiment of subsection \ref{sbs:earthquake}.}


Unless otherwise noted, \blue{the GSO is the adjacency matrix (except for the DCRNN and the GCRN),} the recurrent architectures have a single recurrent layer and the state nonlinearity $\sigma$ is always the $\tanh$ function. The nonlinearities of the GNNs with which we compare our architectures are \blue{also the $\tanh$}. All models are trained using the ADAM algorithm \cite{kingma17-adam} with decaying factors $\beta_1 = 0.9$ and $\beta_2 = 0.999$. We will denote the number of input features by $F_{\bbX}$ and the number of state features by $F_{\bbZ}$, and the number of filter taps in $\ccalA_{\bbS}$ and $\ccalB_{\bbS}$ by $K_{\bbX}$ and $K_{\bbZ}$ respectively. When $\ccalC_{\bbS}$ is a multi-layer GNN, or when comparing against a GNN architecture, the number of features outputted by layer $\ell$ of the GNN is $F_\ell$, and the number of filter taps of this layer is $K_\ell$. When applicable, we denote the number of output features by $F_{\bbY}$.

\subsection{$k$-step prediction: GRNN vs. GNN vs. RNN} \label{sbs:comparison_gnn_rnn}

Let $\ccalG$ be an SBM graph with $N=80$ nodes, $c = 5$ communities, intra-community probability $p_{c_i c_i} = 0.8$ and inter-community probability $p_{c_i c_j} = 0.2$. We write a noisy diffusion process on this graph as 
\begin{equation}
\bbx_t = \bbS \bbx_{t-1} + \bbw_t
\end{equation}
where $\bbx_t \in \reals^N$ is a graph signal, $\bbS \in \reals^{N \times N}$ is the GSO and $\bbw_t \in \reals^N$ is a zero-mean Gaussian noise with temporal variance $\xi^2 = 0.01$ and spatial variance (across nodes) $\eta^2 = 0.01$. The problem of $k$-step prediction consists of estimating $\bbx_{t+k}, \bbx_{t+k+1}, \bbx_{t+k+2}, \ldots$ from $\bbx_{t}, \bbx_{t+1}, \bbx_{t+2}, \ldots$.

We simulate this process for many values of $\bbx_0$ and over multiple time steps, feeding the generated data to three neural network models trained to predict the diffused graph signals $k = 5$ steps ahead. These models are a GRNN, a GNN and a RNN. Using different amounts of training data, our goal is to compare how well these architectures generalize on the test set. 
The GRNN architecture takes in single-feature input sequences ($F_\bbX=1$) and consists of one recurrent layer with $F_\bbZ = 5$ state features and $K_\bbX = K_\bbZ = 5$ filter taps for both the input-to-state and the state-to-state filters. The state features are mapped to the output using a 1-layer GNN with $K_{1} = 1$ and $F_1 = F_{\bbY}=1$, adding up to $155$ parameters. The GNN architecture takes in individual samples $\bbx_t$ from the sequence to predict the sample $k$ steps ahead $\bbx_{t+k}$, and is made up of 2 graph convolutional layers with $F_{\bbX} = 1, F_1 = 8$, $F_{2} = F_\bbY = 1$ and $K_1 = K_2 = 10$, totaling $160$ parameters. Finally, the RNN has one layer and both the input and the output have $N$ features ($F_{\mbox{\scriptsize in}} = F_{\mbox{\scriptsize out}} = N$), corresponding to the nodes of the graph; the state has $F_\bbZ = 1$ features, and the number of parameters is equal to $160$. The fact that all architectures have roughly the same number of parameters is not a coincidence, and was intended to make sure that comparisons are fair.

All architectures were trained by optimizing the $L1$ loss over $10$ epochs with learning rate $10^{-3}$. In the first set of experiments, the size of the training, validation and test sets were $10000$, $2400$ and $200$ sample sequences respectively, and training was done in batches of $100$. The average relative root mean square error (rRMSE) on the test set for 25 Monte-Carlo simulations (corresponding to $5$ different graphs and $5$ different dataset realizations) are presented in Figure \ref{10000}, while the average training rRMSE for each architecture versus the number of training steps is presented in Figure \ref{evolution}. The GRNN outperforms the GNN in almost \blue{3 p.p.}, but the GRNN and RNN achieve roughly the same performance. This makes sense, considering that the solution space searched by the GRNN is a subset of that searched by the RNN. On the other hand, Figure \ref{evolution} shows that the GRNN architecture explores this solution space more efficiently and, consequently, trains faster. This observation is corroborated by the results obtained in the second set of experiments, where the size of the training and validation sets was cut by half. As seen in Figure \ref{5000}, when less training data is available the advantages of the additional structure carried by GRNNs are more perceptible, with the GRNN outperforming the RNN in over 10 percentage points even though their number of parameters is exactly the same. 

%

\begin{figure}[t]
	\centering
	\includegraphics[width=0.84\columnwidth]{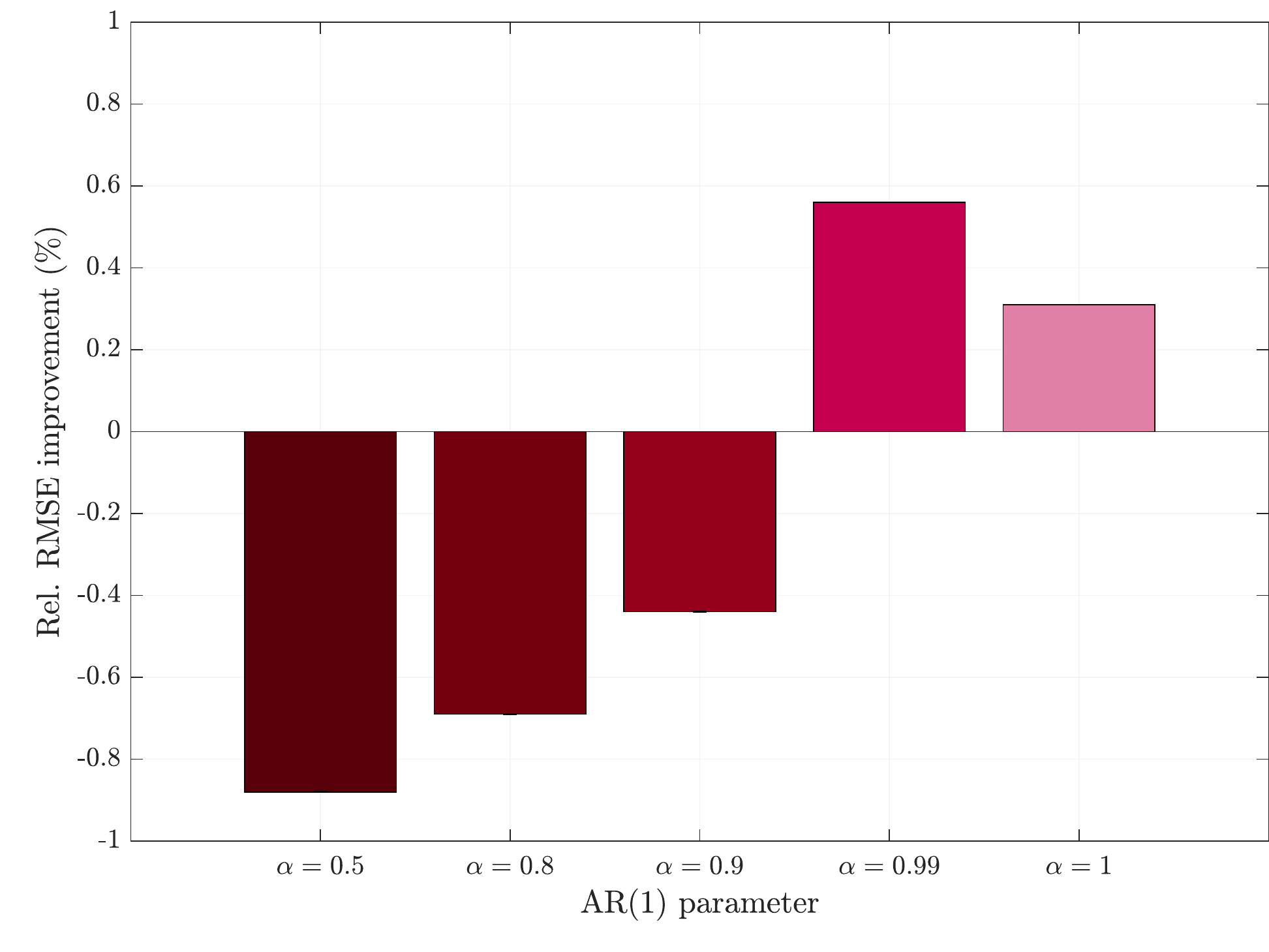} 
	\caption{Relative RMSE improvement of t-GGRNN over GRNN on AR diffusion process [cf. \eqref{eqn:ar}] for various values of $\alpha$.}
	\label{fig:ar}
\end{figure}	

\subsection{$k$-step prediction: AR(1) process and time gating} \label{sbs:ar}

In this experiment, $\ccalG$ is a SBM graph with $N=20$ nodes, $c=2$ communities, intra-community probability $p_{c_ic_i} = 0.8$ and inter-community probability $p_{c_ic_j} = 0.2$. The graph process is an AR(1) process with parameter $0 < \alpha \leq 1$,
\begin{equation} \label{eqn:ar}
\bbx_t = \alpha \bbx_{t-1} + \bbw_t\ 
\end{equation}
where $\bbw_t \in \reals^N$ is a zero-mean Gaussian noise with temporal variance $\xi^2 = 0.01$ and spatial variance $\eta^2 = 0.01$.
When $\alpha$ is close to $0$, this process is weakly correlated in time; when $\alpha \approx 1$, $\bbx_t \approx \bbx_{t-1}$ and the process is strongly correlated.

To assess the advantages of time gating [cf. equation \eqref{eqn:tGGRNN}] in processes with different levels of temporal correlation, we simulate the $k$-step prediction problem for $k=10$ and a range of values of $\alpha$ between $0$ and $1$. The architectures that we compare are a time-gated GRNN and a conventional GRNN with $F_\bbX=1$, $F_\bbZ = 10$ and $K_\bbX = K_\bbZ = 4$. In both GRNNs, the state is mapped to the output using a 1-layer GNN with $F_1 = F_\bbY = 1$ and $K_{\text{out}} = 1$.
The GRNN and the time-gated GRNN are trained by optimizing the L1 loss on $10000$ training samples over $10$ epochs, with learning rate $10^{-3}$ and batch size $100$. The number of samples in the validation and test sets are $2400$ and $200$ respectively, and we report results for 25 Monte Carlo simulations corresponding to 5 different graphs and 5 different datasets per graph. 

The relative RMSE improvement of the time-gated GRNN over the basic GRNN is presented in Figure \ref{fig:ar} for multiple values of $\alpha$. We observe that, for small $\alpha$, the GRNN achieves lower RMSE than the t-GGRNN on average, but, as $\alpha$ increases, its performance improves relative to the GRNN. This can be explained by the fact that, when $\alpha$ is large, there is memory \textit{within the process}, i.e., within $\bbx_t$. Because of this ``built-in memory'', the state $\bbz_t$ is less important than the instantaneous input $\bbx_t$. The forget gate $\check{\ccalG}_t$ thus helps tune this importance, partially shutting off $\bbz_t$ and making it easier to learn long term dependencies from the process itself.

\subsection{$k$-step prediction: graph diffusion and node gating} \label{sbs:diffusion}

Consider an SBM graph $\ccalG$ with $N=20$ nodes, $C=2$ communities, and inter-community and intra-community probabilities $p_{c_ic_j} = 0.8$ and $p_{c_ic_i} = 0.1$. In this experiment, we simulate a graph diffusion process where the GSO $\bbS \in \reals^{N \times N}$ is exponentiated by $\alpha \in (0,1]$,
\begin{equation} \label{eqn:diff}
\bbx_t = \bbS^\alpha \bbx_{t-1} + \bbw_t
\end{equation}
and where $\bbw_t \in \reals^N$ a zero-mean Gaussian with temporal variance $\xi^2 = 0.01$ and spatial variance $\eta^2 = 0.01$. The role of $\alpha$ is to control the process' ``spatial correlation''. The idea is that, the closer $\alpha$ is to $1$, the more correlated the process is across nodes connected by edges of the graph.

To assess the advantages of node gating in spatially correlated graph processes, we simulate this process for a range of values of $\alpha$ and train a GRNN and a node-gated GRNN to predict the diffused signals $k=10$ steps ahead.  Both the GRNN and n-GGRNN take in input sequences with $F_\bbX=1$ input feature and have $F_\bbZ = 10$ state features and $K_\bbX = K_\bbZ = 4$ filter taps. The state features are mapped to the output using a 1-layer GNN with $F_1=F_\bbY=1$ and $K_1 = 1$.
The GRNN and the node-gated GRNN are trained by optimizing the L1 loss on $10000$ training samples over $10$ epochs, with learning rate $10^{-3}$ and batch size $100$. The number of samples in the validation and test sets are $2400$ and $200$ respectively, and we report results for 25 Monte Carlo simulations (5 graphs and 5 datasets per graph). 

The average relative RMSE improvement of the node gated architecture over the basic GRNN architecture are shown in Figure \ref{fig:diffusion} for $\alpha$ in the range $[0.005,0.01, 0.1, 0.2, 0.5, 1]$. 
When $\alpha$ is closer to $0$, the state $\bbz_t$ is more informative than $\bbx_t$, and we see the effect of the input gate shutting off the input; when $\alpha$ is closer to $1$, the input $\bbx_t$ is more informative, and it is now the forget gate that shuts off the state. In the mid-range, the effects of both gates are combined, yielding the largest performance improvements of the n-GGRNN over the GRNN.

\subsection{$k$-step prediction: covariance graphs and edge gating} \label{sbs:covariance}

%

\begin{figure}[t]
	\centering
	\includegraphics[width=0.85\columnwidth]{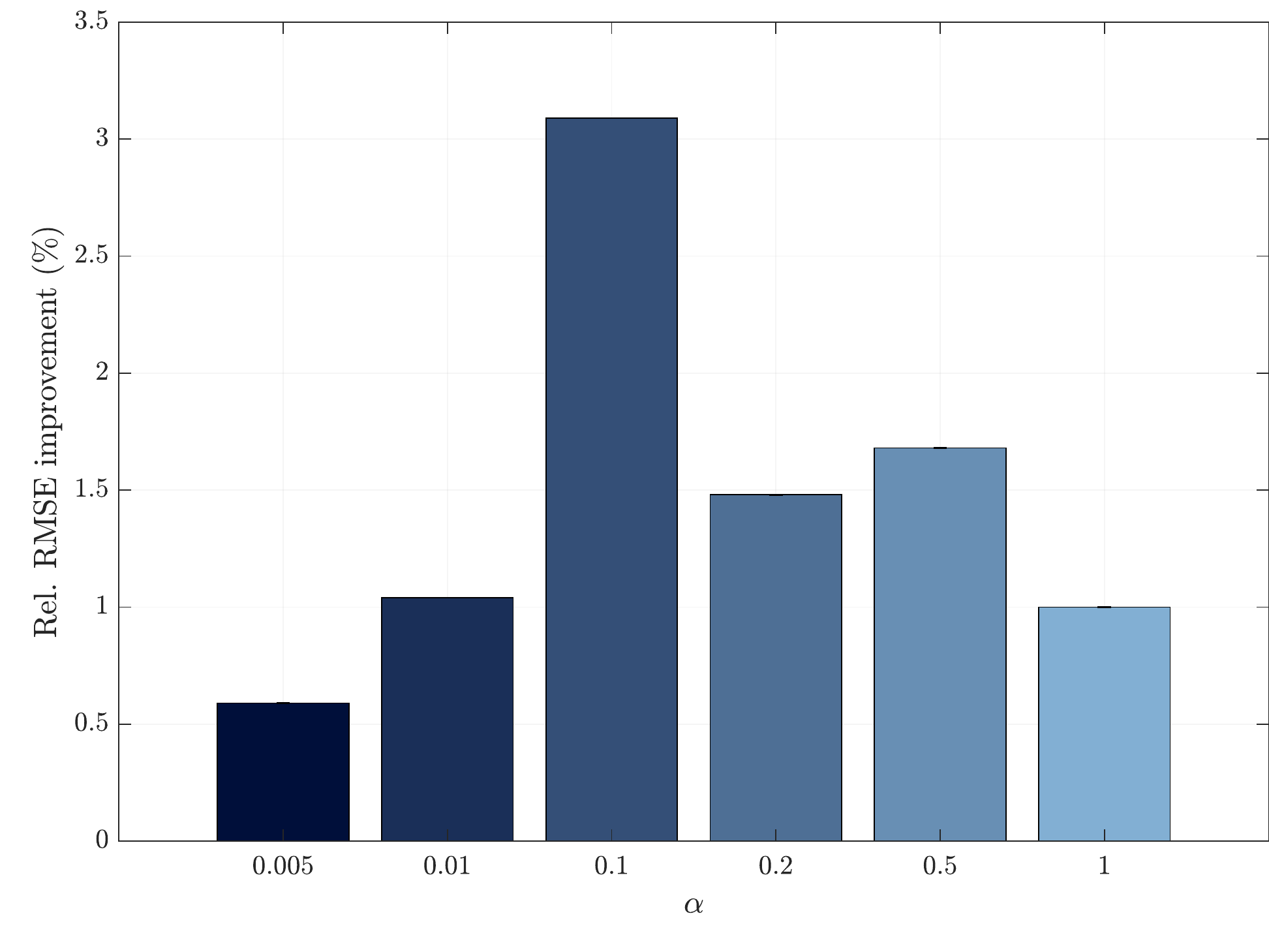} 
	\caption{Relative RMSE improvement of n-GGRNN over GRNN on graph diffusion process [cf. \eqref{eqn:diff}] for various values of $\alpha$.}
	\label{fig:diffusion}
\end{figure}	

In this experiment, the graph $\ccalG$ is a $n$-nearest-neighbor covariance graph with $N=20$ nodes, and the diffusion process is a simple graph diffusion given by
\begin{equation}
\bbx_t = \bbS \bbx_{t-1} + \bbw_t
\end{equation}
where $\bbw_t \in \reals^N$ a zero-mean Gaussian with temporal variance $\xi^2 = 0.01$ and spatial variance $\eta^2 = 0.01$. 
The GSO $\bbS$ is the sample covariance of the $\bbx_0$ in the training set, which consist of $10000$ samples taken from a multivariate normal with mean $\boldsymbol{\mu}= \boldsymbol{1}$ and true covariance matrix $[\Sigma]_{ii} = 3$, $[\Sigma]_{ij} = 1$ for $1 \leq i,j \leq 20$, $i \neq j$. This procedure yields complete graphs and so, to assess the effects of edge gating in graphs with different levels of connectivity, we set each node's maximum number of neighbors to $n$, varying $n$ between $5$ and $20$. 

The architectures we compare are an edge-gated GRNN and a basic GRNN, both with $F_\bbX = 1$, $F_\bbZ = 10$, $K_\bbX = K_\bbZ = 4$ and followed by an output GNN with one layer, $K_\bbY = 1$ and $F_1 = F_\bbY = 1$. These architectures are trained by optimizing the L1 loss on $10000$ training samples over $10$ epochs, with learning rate $10^{-3}$ and batch size $100$. The number of samples in the validation and test sets are $2400$ and $200$ respectively, and we report results for 25 Monte Carlo realizations (5 graphs and 5 datasets per graph).

The relative RMSE improvement of the e-GRNN over the GRNN is presented in Figure \ref{fig:cov} for each value of $n$. When $n$ is small, the e-GGRNN produces a larger relative test RMSE than the non-gated GRNN on average, but as $n$ increases this behavior gradually shifts and the effects of edge gating can be perceived. In particular, we observe that for $n=15$ and $n=20$ the edge-gated GRNN outperforms the GRNN.

\subsection{Earthquake epicenter estimation} \label{sbs:earthquake}

In this problem, we use seismic wave data from the Geonet database \cite{geonet} to predict the region of origin of $2289$ earthquakes registered between June 17, 2019 and July 17, 2019 in New Zealand. The graph $\ccalG$ is a 3-nearest neighbor \blue{directed} network constructed from the coordinates of $N=59$ seismographs, and the input data consists of \blue{$10\text{s}$, $20\text{s}$ and $30\text{s}$ seismic wave readings sampled at $2$ Hz and registered immediately before the seisms.}
The prediction space $\ccalY$ are the $C=11$ geographic regions of New Zealand, each of which is matched with a class label.

%

\begin{figure}[t]
	\centering
	\includegraphics[width=0.86\columnwidth]{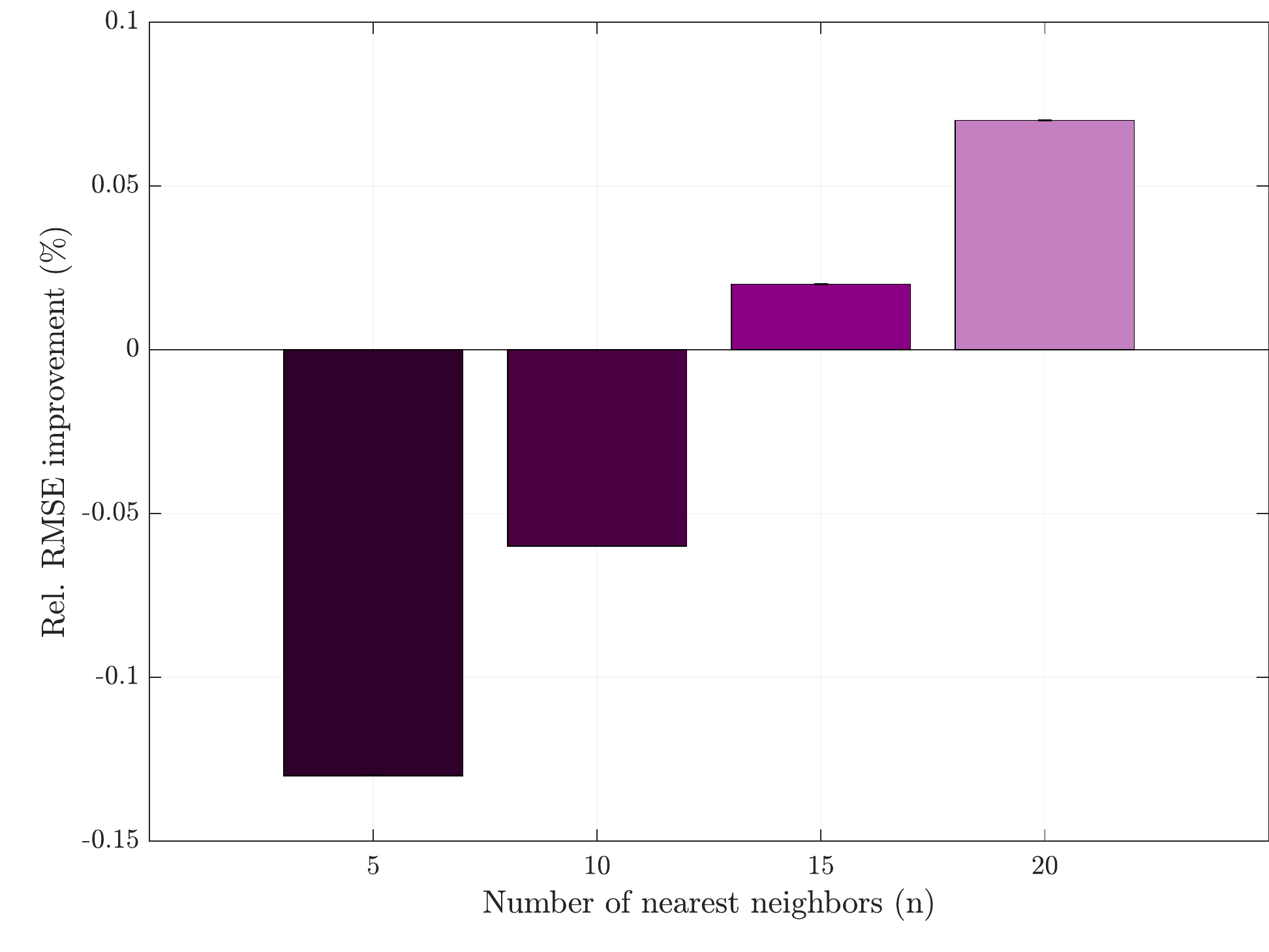} 
	\caption{Relative RMSE improvement of e-GGRNN over GRNN on covariance graphs with 5, 10, 15 and 20 nearest neighbors.}
	\label{fig:cov}
\end{figure}	
\blue{
\begin{table}[t]
\centering
\begin{tabular}{c|c}  \hline
\multicolumn{2}{c}{Number of parameters in the $10\text{s}$ case} \\ \hline
 GNN & $1$ layer, $F_\bbX = 20, F_1=6, K_1=3$ \\
 RNN & $1$ layer, $F_\bbX=1, F_\bbZ=6$ \\
 GRNN w/o gates & $F_\bbX = 1, F_\bbZ=8, K_\bbX = 5, K_\bbZ = 5$ \\ \hline
 \multicolumn{2}{c}{Number of parameters in the $20\text{s}$ case} \\ \hline
 GNN & $1$ layer, $F_\bbX = 40, F_1=8, K_1=4$ \\
 RNN & $1$ layer, $F_\bbX=1, F_\bbZ=20$ \\
 All GRNNs + \cite{li2017diffusion} & $F_\bbX = 1, F_\bbZ=16, K_\bbX = 5, K_\bbZ = 5$ \\ \hline
  \multicolumn{2}{c}{Number of parameters in the $30\text{s}$ case} \\ \hline
   All GRNNs + \cite{li2017diffusion} & $F_\bbX = 1, F_\bbZ=32, K_\bbX = 5, K_\bbZ = 5$ \\ \hline
\end{tabular}
\caption{\blue{Number of parameters of the GNN, RNN and all GRNN architectures in the earthquake epicenter estimation experiment for $10\text{s}$, $20\text{s}$, and $30\text{s}$-long seismic waves.}}
\label{table:params-earthquake10}
\end{table}
}

\blue{Two experimental scenarios are considered. In the first, we compare GRNNs with RNNs and GNNs to analyze the advantages of GRNNs in a real-world setting. In the second, the basic GRNN architecture is compared with time, node and edge-gated GRNNs, and with the gated DCRNN architecture from \cite{li2017diffusion} to assess the advantages of gating.}
The hyperparameters of these architectures are presented in Table \ref{table:params-earthquake10} for the $10\text{s}$, $20\text{s}$ and \blue{$30\text{s}$} cases, and were set to ensure roughly the same number of parameters in the inner layers of these architectures (except for the time, node, and edge-gated architectures, which have one additional GRNN per gate).
\blue{In all experiments, we report the average test accuracy of $10$ Monte Carlo simulations using $1648$ earthquakes for training, $412$ for validation and $229$ for testing, and optimize the cross-entropy loss over $40$ epochs with learning rate $5\times 10^{-5}$ and batch size $50$. The cross-entropy loss was adjusted for class imbalance by adding class weights inversely proportional to the class sizes.}

\blue{\subsubsection{GRNN vs. GNN vs. RNN}
The average and maximum test accuracy achieved by the GRNN, the GNN and the RNN are reported in Table \ref{table:earthquake10} for both the $10\text{s}$ and $20\text{s}$ input signals. Note that, in the $10\text{s}$ case, the GRNN outperforms the RNN on average, but achieves virtually the same performance as the GNN. Also note that while the RNN has best overall accuracy for one of the data splits, it has the largest variance. A possible reason for this inconsistent performance is the fact that the exchanges of node feature information in the RNN do not always match the communication structure defined by the edges of the seismograph network.
When we swap the input data for longer sequences (the $20\text{s}$ waves), the GRNN is able to retain the same performance as in the $10\text{s}$ case. Meanwhile, the accuracy of both the GNN and the RNN degrades. In the case of the GNN, the drop in accuracy is expected as the input sequence has doubled in length and the GNN is not recurrent. In the case of the RNN, it corroborates the observation made in subsection \ref{sbs:comparison_gnn_rnn} that, even if the RNN has more representative power than the GRNN, it searches the representation space less efficiently because it lacks structural information about the graph. This presents a disadvantage in ``harder'', real-world problems, where there is a limited amount of training data, i.e., where data cannot be synthetically generated.}

\begin{table}[t]
\centering
\begin{tabular}{c|c|ccc}  \hline
  \multicolumn{2}{c|}{} & GNN & RNN & GRNN \\ \hline
 \multirow{3}{*}{$10\text{s}$} & \multirow{2}{*}{mean} & $13.4$\% & $11.1$\% & $\mathbf{13.8}$\% \\
 &  & $\pm 1.8$\% & ${\pm 8.3}$\% & $\mathbf{\pm 2.7}$\% \\ \cline{2-5}
 & max & $16.2$\% & $\mathbf{27.9}$\% & $19.2$\% \\ \hline
  \multirow{3}{*}{$20\text{s}$} & \multirow{2}{*}{mean} & $12.9$\% & $8.4$\% & $\mathbf{14.1}$\% \\
   &  & $\pm 2.0$\% & ${\pm 3.8}$\% & $\mathbf{\pm 3.1}$\% \\ \cline{2-5}
 & max & $15.7$\% & $15.3$\% & $\mathbf{19.7}$\% \\ \hline
\end{tabular}
\caption{\blue{Average earthquake region prediction accuracy (\%) achieved by the GNN, RNN, and GRNN for $10\text{s}$ and $20\text{s}$-long seismic waves.}}
\label{table:earthquake10}
\end{table}

\blue{
\subsubsection{Gating}
The average and maximum test accuracy achieved by the non-gated GRNN, the time, node and edge-gated GRNNs, and the DCRNN from \cite{li2017diffusion} are reported in Table \ref{table:earthquake10} for both the $20\text{s}$ and $30\text{s}$ seismic waves. In the $20\text{s}$ case, all architectures achieve roughly the same performance on average, and only the edge-gated GRNN and the DCRNN achieve a better overall performance than the GRNN in their best data splits. On the other hand, in the $30\text{s}$ case the e-GRNN is clearly best, outperforming the GRNN in over 3 p.p. and classifying more than a third of the earthquakes in the test set correctly in the best data split. The time-gated and the node-gated GRNN do not improve upon the non-gated GRNN on average, but the node-gated GRNN exhibits a smaller variance. We observe the same pattern for the DCRNN, which uses the random walk matrix as the GSO and also contains node gates. These results show that the choice of GSO has no significant effect in performance. They also agree with intuition---the more a gating strategy encodes the graph structure, and the more degrees of freedom it has, the better it should perform with respect to the non-gated GRNN.
}

\begin{table}[t]
{\footnotesize
\centering
\begin{tabular}{c|c|ccccc}  \hline
  \multicolumn{2}{c|}{} & GRNN & t-GRNN & n-GRNN  & e-GRNN & \cite{li2017diffusion}\\ \hline
 \multirow{3}{*}{$20\text{s}$} & \multirow{2}{*}{mean} & $14.9$\% & $14.3$\% & ${13.7}$\% & ${15.4}$\%   & $\mathbf{15.6}$\% \\
 &  & $\pm 2.8$\% & ${\pm 1.7}$\% & ${\pm 2.4}$\% & ${\pm 3.6}$\% & $\mathbf{\pm 2.7}$\% \\ \cline{2-7}
 & max & $19.7$\% & ${16.2}$\% & $17.9$\% & $\mathbf{21.0}$\% & $20.1$\%\\ \hline
  \multirow{3}{*}{$30\text{s}$} & \multirow{2}{*}{mean} & $16.1$\% & $15.4$\% & ${15.7}$\% & $\mathbf{19.7}$\% & $15.3$\%\\
   &  & $\pm 4.5$\% & ${\pm 4.7}$\% & ${\pm 2.4}$\% & $\mathbf{\pm 5.3}$\% & ${\pm 2.2}$\% \\ \cline{2-7}
 & max & $25.3$\% & $27.9$\% & ${19.2}$\% & $\mathbf{32.2}$\% & $19.2$\%\\ \hline
\end{tabular}
\caption{\blue{Average earthquake region prediction accuracy (\%) achieved by the GRNN, t-GGRNN, n-GGRNN, e-GGRNN and DCRNN \cite{li2017diffusion} for $20\text{s}$ and $30\text{s}$-long seismic waves.}}
\label{table:earthquake20}
}
\end{table}


\blue{

\subsection{Traffic forecasting in Los Angeles} \label{sbs:traffic}

The goal of this experiment is to optimize the GRNN architecture and its gated variations to predict traffic speeds using the METR-LA dataset \cite{jagadish2014big}, which is a commonly used dataset for benchmark in traffic forecasting. For this reason, we also compare our architectures with the DCRNN from \cite{li2017diffusion} and the GCRN from \cite{seo2018structured}. 
The METR-LA dataset is a collection of speed readings measured between March and June 2012 by $N=207$ speed radars in the Los Angeles metropolitan area, and aggregated in 5-minute windows. The speed sensors are connected through an undirected geometric graph with $N$ nodes, which measures their pairwise road network distance. 

In our experiments, we use 20\% of the data split provided by \cite{li2017diffusion}, which utilizes 70\% of the data for training, 10\% for validation and 20\% for testing\footnote{Available at \url{https://github.com/liyaguang/DCRNN}.}. 
All simulated architectures consist of $L=1$ recurrent layer with $F_{\bbX}=1$ input feature, $F_{\bbZ}=24$ state features and $K_{\bbX}=K_{\bbZ}=5$ filter taps. The state features are mapped to the output using a 1-layer GNN with $K=1$ and $F_1=F_{\bbY}=1$ to predict the traffic speeds. We optimize the L1 loss over $40$ epochs, with learning rate $0.01$ and batch size $400$. The rRMSEs achieved on the test set by each architecture are presented in Table \ref{table:traffic}. All architectures achieve roughly the same performance, except for the time-gated GRNN, which has the worst performance by and large. This shows that time gating can be somewhat contrived for big graphs and graph processes with large spatial correlation, which is the case of traffic data. On the other hand, no spatial gating mechanism or graph recurrent architecture seems to be better than the other in this problem in particular. Edge and node gated architectures achieve similar test error and do not improve upon the GRNN. The node gated architectures, including the n-GRNN, DCRNN \cite{li2017diffusion} and GCRN \cite{seo2018structured}, are also pretty much equivalent, showing that the choice of GSO (adjacency, Laplacian or random walk matrix) has no bearing on each architecture's performance in this problem.
 
\begin{table}[t]
{
\centering
\begin{tabular}{ccccccc}  \hline
GRNN & t-GRNN & n-GRNN & e-GRNN & \cite{li2017diffusion} & \cite{seo2018structured} \\ \hline
$26.38$\% & $33.55$\% & $26.39$\% & $26.95$\% & $26.67$\% & $\mathbf{26.25}$\% \\ \hline
\end{tabular}
\caption{\blue{Relative RMSE (\%) achieved on the test set by the GRNN, t-GGRNN, n-GGRNN, e-GGRNN, DCRNN \cite{li2017diffusion} and GCRN \cite{seo2018structured} on the METR-LA dataset.}}
\label{table:traffic}
}
\end{table}

\subsection{Epidemic tracking on a friendship network} \label{sbs:epidemic}

In this experiment, we compare GRNNs, gated GRNNs and the gated graph recurrent architectures from \cite{li2017diffusion} and \cite{seo2018structured} in a binary node classification problem modeling the spread of an epidemic on a high school friendship network. 
The graph is an unweighted and undirected $134$-node network built from real data collected at a high school in Marseilles, France, in December 2013 \cite{mastrandrea2015contact}. The friendships are reported as directed links, but we symmetrize the network to model the spread of the disease in both directions if there is a contact between friends. To make sure the graph is connected, we also remove all of its isolated nodes.
As for the epidemic data, it is generated by using the SIR model to simulate the spread of an infectious disease on the friendship network \cite{fournet2017estimating}. The disease is first recorded on day $t=0$, when each individual node is infected with probability $p_{\mbox{\tiny seed}}=0.05$. On the days that follow, an infected student can then spread the disease to their susceptible friends with probability $p_{{\mbox{\tiny inf}}}=0.3$ each day. Infected students become immune after $4$ days, at which point they can no longer spread or contract the disease. 

Given the state of each node at some point in time (susceptible, infected or recovered), the binary node classification problem is to predict whether each node in the network will have the disease (i.e., be infected) $k=8$ days ahead. We train 6 models to solve this problem: a GRNN; time, node and edge-gated GRNNs; the DCRNN \cite{li2017diffusion}; and the GCRN \cite{seo2018structured}. All contain $L=1$ recurrent layer with $F_{\bbX}=1$ input feature, $F_{\bbZ}=12$ state features and $K_{\bbX}=K_{\bbZ}=5$ filter taps. The state features are mapped to the output using a 1-layer GNN with $K=1$ and $F_1=F_\bbY = 2$ features corresponding to the number of classes, followed by a softmax layer. The models are trained by optimizing the F1-score over $10$ training epochs, with learning rate $5\times 10^{-4}$ and batch size $100$. We report results for 10 Monte-Carlo realizations split between $1000$ samples for training, $120$  for validation and $200$ for testing.

The average F1-score, precision and recall achieved by each architecture on the test set are presented in Table \ref{table:epidemic}. We observe that the time and node-gated GRNNs and the GCRN \cite{seo2018structured}, which is also a node-gated architecture, achieve a F1-score that is $\sim 2$ p.p. higher than the F1-score of the GRNN. Meanwhile, the edge-gated GRNN and the DCRNN \cite{li2017diffusion} perform worse than the GRNN. The better performance of the node-gated architectures (except for the DCRNN) could be explained by the fact that node gating allows stopping information flows from nodes that have reached the recovered state and that, as such, do not affect disease transmission. However, this is only true for the n-GRNN and the GCRN \cite{seo2018structured}, which use the normalized adjacency and the normalized Laplacian respectively as their GSOs. In the case of the DCRNN, the benefit of node gating is cancelled by the fact that its GSO ---the random walk matrix--- penalizes the weights of edges associated with nodes with high degrees. This makes it harder to model the spread of the disease along these edges with graph convolutions, given that disease transmission is stochastic and only depends on there being an edge between two students, but not on the edge weight. A possible explanation for the good performance of the time-gated GRNN is the fact that, after a few time steps, the nodal states become more homogeneous: susceptible nodes give room to infected nodes at first, which later become recovered. 
As for the edge-gated GRNN, its worse performance with respect to other gating strategies can be related to the fact that stopping information flows from inactive nodes might require gating a large number of edges, which is more difficult than gating a few nodes (node gating) or the entire exchange (time gating). Finally, note that the t-GRNN, the n-GRNN and the GCRN \cite{seo2018structured} also achieve the highest recall among all architectures and therefore minimize false negatives, which is especially important when monitoring the spread of an infectious disease.

\begin{table}[t]
{\centering
\begin{tabular}{l|c|c|c} \hline
       & F1-score & Precision & Recall   \\ \hline
GRNN   & $0.782 \pm 0.044$    & $0.803 \pm 0.004$     & $0.958 \pm 0.076$   \\ 
t-GRNN & $\mathbf{0.805 \pm 0.003}$    & $0.802 \pm 0.004$    & $1.000 \pm 0.000$  \\ 
n-GRNN & $0.800 \pm 0.014$   & $0.802 \pm 0.004$     & $0.990 \pm 0.030$   \\ 
e-GRNN   & $0.682 \pm 0.247$    & $0.841 \pm 0.079$     & $0.838 \pm 0.325$   \\ 
\cite{li2017diffusion} & $0.745 \pm 0.179$    & $0.821 \pm 0.060$ & $0.920 \pm 0.240$ \\ 
\cite{seo2018structured}  & $\mathbf{0.805 \pm 0.003}$     & $0.802 \pm 0.004$      & $0.999 \pm 0.004$  \\ \hline
\end{tabular}
\caption{\blue{F1-score, precision and recall achieved by the GRNN, t-GGRNN, n-GGRNN, e-GGRNN, DCRNN \cite{li2017diffusion} and GCRN \cite{seo2018structured} on the epidemic modeling problem.}}
\label{table:epidemic}
}
\end{table}

}





\section{Conclusions} \label{sec:conclusions}

We have introduced a GSP-oriented GRNN framework tailored to learning problems involving graph processes. Merging the recurrent architecture of RNNs with graph convolutional layers, GRNNs are able to take both the sequential structure of data and the underlying graph topology into account. We have shown that GRNNs are Lipschitz stable to graph perturbations with a Lipschitz constant that is polynomial on the length of the graph processes considered. The GRNN architecture was also extended to include three gating strategies: (i) time gating, which enables encoding long term time dependencies without assigning them exponentially smaller/larger weights; (ii) node gating, in which each node has a gate and the graph structure is leveraged to control spatial dependencies on the graph; and (iii) edge gating, which achieves the same purpose but by assigning gates to edges instead. \blue{The advantages of these architectures were demonstrated in both synthetic and real-world problems, where they were also compared with gated graph recurrent architectures from the literature ---the DCRNN and the GCRN--- that fit the GRNN framework discussed in Sections \ref{sec:GRNN} and \ref{sec:gatedGRNN}. We observe that GRNN architectures outperform GNNs and RNNs in both settings and are more general than the architectures from \cite{seo2018structured,li2017diffusion}. We also note that the different gating strategies can sensibly improve GRNN performance in problems with long term temporal dependencies and long range spatial dependencies on the graph.} 


\appendices


\section{Proof of Proposition 1} \label{sec:appendixA}
\begin{proof}[Proof of Proposition \ref{prop:permutationEquivariance}]
Since the permutation matrix $\bbP \in \ccalP$ is orthogonal, we have $\bbP^\Tr\bbP=\bbP\bbP^\Tr$, which implies
\begin{equation}
\tbS^k = (\bbP^\Tr\bbS\bbP)^k = \bbP^\Tr\bbS^k\bbP\text{.}
\end{equation}
Writing $\bbA(\tbS)$ as in \eqref{eqn:graphConv}, we get
\begin{equation}
\bbA(\tbS) = \bbP^\Tr \bbA(\bbS) \bbP
\end{equation}
and so applying $\bbA(\tbS)$ to $\tbx=\bbP^\Tr\bbx$ yields
\begin{equation}
\bbA(\tbS)\tbx = \bbP^\Tr \bbA(\bbS) \bbP \bbP^\Tr \bbx = \bbP^\Tr \bbA(\bbS) \bbx \text{.}
\end{equation}
Graph convolutions are thus permutation equivariant. Using \eqref{eqn:GRNNhidden}, we can then write $\tbz_t$ as
\begin{align}
\tbz_t&=\sigma(\bbA(\tbS)\tbx_t+\bbB(\tbS)\tbz_{t-1})\\
&= \sigma(\bbP^\Tr\bbA(\bbS)\bbx_t+\bbP^\Tr\bbB(\bbS)\bbz_{t-1}) \\
&=\bbP^\Tr\sigma(\bbA(\bbS)\bbx_t+\bbB(\bbS)\bbz_{t-1}) = \bbP^\Tr \bbz_t
\end{align}
where the second-to-last equality follows from the fact that $\sigma$ is pointwise and hence permutation equivariant. 
Since $\rho$ is also pointwise, by a similar reasoning we have $\tby_t = \rho(\bbC(\tbS)\tbz_t)=\bbP^\Tr \rho(\bbC(\bbS)\bbz_t) = \bbP^\Tr\bby_t$.
\end{proof}

\section{Proof of Theorem 1} \label{sec:appendixB}

\begin{lemma} \label{lemmaStab}
Let $\bbS=\bbV\bbLam\bbV^\Hr$ and $\tilde{\bbS}$ be graph shift operators. Let $\bbE=\bbU\bbM\bbU^\Hr \in \ccalE(\bbS,\tilde{\bbS})$ be a relative perturbation matrix [cf. Definition \ref{def:errorSet}] whose norm is such that
\begin{equation*}
d(\bbS,\tilde{\bbS}) \leq \|\bbE\| \leq \varepsilon\text{.}
\end{equation*}
For an integral Lipschitz filter [cf. Definition \ref{def:integralLipschitz}] with integral Lipschitz constant $C$, the operator distance modulo permutation between filters $\bbH(\bbS)$ and $\bbH(\tilde{\bbS})$ satisfies
\begin{equation}
\|\bbH(\bbS)-\bbH(\tilde{\bbS})\|_{\ccalP} \leq 2C\left( 1 + \delta\sqrt{N}\right)\varepsilon + \ccalO(\varepsilon^2)
\end{equation}
with $\delta := (\|\bbU-\bbV\|_2 + 1)^2 - 1$ standing for the eigenvector misalignment between shift operator $\bbS$ and error matrix $\bbE$.
\end{lemma}
\begin{proof}
See \cite[Theorem 3]{Gama19-Stability}.
\end{proof}

\begin{proof}[Proof of Theorem \ref{thm:stability}]
Without loss of generality, assume $\bbP=\bbI$ in \eqref{eqn:graphDistance} and write $\tilde{\bbS}=\bbS + \bbE \bbS + \bbS \bbE^{\Tr}$, $\tbA = \bbA(\tbS)$, $\tbB = \bbB(\tbS)$ and $\tbC = \bbC(\tbS)$. 
From \eqref{eqn:GRNNoutput}, we can write 
\begin{align} \label{eqn:output_diff1}
\|\bby_t - \tilde{\bby}_t\| = \|\rho(\bbC\bbz_t)-\rho(\tbC\tbz_t)\| \leq \|\bbC\bbz_t-\tbC\tbz_t\|
\end{align}
since $\rho(\cdot)$ is normalized Lipschitz. Adding and subtracting $\bbC\tbz$ on the right-hand side of \eqref{eqn:output_diff1}, and using both the triangle and Cauchy-Schwarz inequalities, we get
\begin{align} \label{eqn:output_diff1.1}
\|\bby_t - \tilde{\bby}_t\| \leq \|\bbC\|\|\bbz_t-\tbz_t\| + \|\bbC-\tbC\|\|\tbz_t\|\text{.}
\end{align}
The norm of $\bbC$ is assumed bounded, and Lemma \ref{lemmaStab} gives a bound to $\|\bbC-\tbC\|$. Using \eqref{eqn:GRNNhidden}, we can write
\begin{align}
\|\bbz_t-\tbz_t\| &=\|\sigma(\bbA\bbx_t+\bbB\bbz_{t-1})-\sigma(\tbA\bbx_t+\tbB\tbz_{t-1})\| \\
&\leq \|\bbA\bbx_t+\bbB\bbz_{t-1}-(\tbA\bbx_t+\tbB\tbz_{t-1})\| \\
&\leq \|\bbA-\tbA\|\|\bbx_t\| + \|\bbB\bbz_{t-1}-\tbB\tbz_{t-1}\| \label{eqn:state_diff1}
\end{align}  
where the first inequality follows from the fact that $\sigma(\cdot)$ is also normalized Lipschitz and the second from the triangle and Cauchy-Schwarz inequalities respectively. The norm difference $\|\bbA-\tbA\|$ is bounded by Lemma \ref{lemmaStab} and $\|\bbx_t\| \leq \|\bbx\|$ for all $t$, so we move onto deriving a bound for the second summand of \eqref{eqn:state_diff1}. We rewrite it as
\begin{align} \label{eqn:state_diff2}
\begin{split}
\|\bbB\bbz_{t-1}+&\bbB\tbz_{t-1}-\bbB\tbz_{t-1}-\tbB\tbz_{t-1}\| \\ 
&\leq \|\bbB\|\|\bbz_{t-1}-\tbz_{t-1}\|+\|\bbB-\tbB\|\|\tbz_{t-1}\| 
\end{split}
\end{align}
which results in a recurrence relationship between $\|\bbz_t-\tbz_t\|$ and $\|\bbz_{t-1}-\tbz_{t-1}\|$. Expanding this recurrence, we obtain
\begin{align*}
\|\bbz_t-\tbz_t\| &\leq \sum_{i=0}^{t-1}\|\bbB\|^i\|\bbA-\tbA\|\|\bbx\| \\
&+ \|\bbB\|^t\|\bbz_0-\tbz_0\| + \|\bbB-\tbB\|\sum_{i=1}^t \|\tbz_{t-i}\| \\ 
&\leq \sum_{i=0}^{t-1}\|\bbB\|^i\|\bbA-\tbA\|\|\bbx\| + \|\bbB-\tbB\|\sum_{i=0}^{t-1} \|\tbz_{i}\|
\end{align*}
where the second inequality follows from $\bbz_0=\tbz_0$. Now it suffices to bound $\|\bbz_i\|$ for any given $i>0$. Writing $\bbz_t$ as in \eqref{eqn:GRNNhidden} and observing that, because $\sigma(\cdot)$ is normalized Lipschitz and $\sigma(0)=0$, $|\sigma(x)|<|x|$,  we can use the triangle and Cauchy-Schwarz inequalities to write
\begin{align} \label{eqn:bound_nongated_state}
\begin{split}
\|\bbz_i\| &\leq \|\bbA\|\|\bbx_i\|+\|\bbB\|\|\bbz_{i-1}\| \leq \ldots\\
&\leq \sum_{j=0}^{i-1}\|\bbB\|^j\|\bbA\|\|\bbx\|+\|\bbB\|^i\|\bbz_0\| 
\end{split}
\end{align}
for $i>0$.
Substituting this in \eqref{eqn:state_diff2}, we get
\begin{align} \label{eqn:state_diff3}
\begin{split}
&\|\bbz_t-\tbz_t\| \leq \|\bbA-\tbA\|\|\bbx\|\sum_{i=0}^{t-1}\|\bbB\|^i \\
&+ \|\bbB-\tbB\|\bigg(\|\tbA\|\|\bbx\|\sum_{i=1}^{t-1} \sum_{j=0}^{i-1}\|\tbB\|^j+\|\bbz_0\|\sum_{i=0}^{t-1}\|\tbB\|^i\bigg) \text{.} 
\end{split}
\end{align}
Finally, substituting equations \eqref{eqn:state_diff2} and \eqref{eqn:state_diff3} in \eqref{eqn:output_diff1.1} gives
\begin{align*}
&\|\bby_t-\tby_t\| \leq \|\bbC\|\bigg[\|\bbA-\tbA\|\|\bbx\|\sum_{i=0}^{t-1}\|\bbB\|^i \\
&+ \|\bbB-\tbB\|\bigg(\|\tbA\|\|\bbx\|\sum_{i=1}^{t-1} \sum_{j=0}^{i-1}\|\tbB\|^j+\|\bbz_0\|\sum_{i=0}^{t-1}\|\tbB\|^i\bigg)\bigg]  \\
&+ \|\bbC-\tbC\|\bigg[\sum_{i=0}^{t-1}\|\tbB\|^i\|\tbA\|\|\bbx\|+\|\tbB\|^t\|\bbz_0\|\bigg] \text{.}
\end{align*}
This expression can be simplified by applying Lemma \ref{lemmaStab} to the norm differences $\|\bbA-\tbA\|$, $\|\bbB-\tbB\|$ and $\|\bbC-\tbC\|$, and by recalling that $\|\bbA\|=\|\bbB\|=\|\bbC\|=1$, $\|\bbx\|=1$ and $\bbz_0=\boldsymbol{0}$. Denoting $C = \max\{C_{\bbA},C_{\bbB},C_{\bbC}\}$ the maximum filter Lipschitz constant, we recover \eqref{eqn:stability} with $\bbP=\bbI$,
\begin{equation} 
\| \bby_{t} - \tby_{t} \| \leq C(1+\sqrt{N} \delta)(t^2+3t) \varepsilon\ + \ccalO(\varepsilon^{2})
\end{equation}
which completes the proof. 
\end{proof}

\section{Proof of Theorem 2} \label{sec:appendixC}
\blue{
\begin{proof}[Proof of Theorem \ref{thm:GGRNNstability}]
Without loss of generality, we will assume $\bbP = \bbI$ in \eqref{eqn:graphDistance} and write $\tbS=\bbS+\bbE\bbS+\bbS\bbE^\Tr$, $\tbA=\bbA(\tbS)$, $\tilde{\hbA}=\hbA(\tbS)$, $\tilde{\cbA}=\cbA(\tbS)$, $\tbB=\bbB(\tbS)$, $\tilde{\hbB}=\hbB(\tbS)$, $\tilde{\cbB}=\cbB(\tbS)$ and $\tbC=\bbC(\tbS)$.
We also denote the input and forget gate operators whose parameters $\hat{\bbtheta}$ and $\check{\bbtheta}$ have been perturbed by
\begin{align*}
\begin{split}
\hcalQ_{\tilde{\hat{\bbtheta}}} = \tilde{\hcalQ} \quad \mbox{and} \quad \kcalQ_{\tilde{\check{\bbtheta}}} = \tilde{\kcalQ}.
\end{split}
\end{align*}

By the same reasoning used in equations \eqref{eqn:output_diff1} and \eqref{eqn:output_diff1.1} of the proof of Theorem \ref{thm:stability}, we start by bounding $\|\bby_t-\tby_t\|$ as
\begin{align} \label{eqn:thm2_output_diff1.1}
\|\bby_t - \tilde{\bby}_t\| \leq \|\bbC\|\|\bbz_t-\tbz_t\| + \|\bbC-\tbC\|\|\tbz_t\|\text{.}
\end{align}
Using \eqref{eqn:gatingGeneric} and AS\ref{as2}, we can write
\begin{align} \label{eqn:thm2_state_diff1}
\begin{split}
\|\bbz_t-\tbz_t\| =&\ \|\sigma(\hcalQ(\bbA\bbx_t)+\kcalQ(\bbB\bbz_{t-1}))\\
&-\sigma(\tilde{\hcalQ}(\tbA\bbx_t)+\tilde{\kcalQ}(\tbB\tbz_{t-1}))\| \\
&\leq \|\hcalQ(\bbA\bbx_t)-\tilde{\hcalQ}(\tbA\bbx_t)\| \\
&+\|\kcalQ(\bbB\bbz_{t-1})-\tilde{\kcalQ}(\tbB\tbz_{t-1})\| 
\end{split}
\end{align} 
where the bound follows from the triangle inequality.
Focusing on the first term on the right-hand side of \eqref{eqn:thm2_state_diff1}, we apply the triangle once again to get
\begin{align} \label{eqn:thm2_state_diff1.1}
\begin{split}
\|\hcalQ&(\bbA\bbx_t)-\tilde{\hcalQ}(\tbA\bbx_t)\| = \ldots\\
&= \|\hcalQ(\bbA\bbx_t)+\tilde{\hcalQ}(\bbA\bbx_t)-\tilde{\hcalQ}(\bbA\bbx_t)-\tilde{\hcalQ}(\tbA\bbx_t)\| \\
&\leq \|\hcalQ(\bbA\bbx_t)-\tilde{\hcalQ}(\bbA\bbx_t)\| +\|\tilde{\hcalQ}(\bbA\bbx_t)-\tilde{\hcalQ}(\tbA\bbx_t)\|.
\end{split}
\end{align}
Using the Cauchy-Schwarz inequality and AS\ref{as5},
\begin{align} \label{eqn:1.1}
\begin{split}
\|\hcalQ(\bbA\bbx_t)-\tilde{\hcalQ}(\bbA\bbx_t)\| \leq \|\hcalQ-\tilde{\hcalQ}\|\|\bbA\|\|\bbx_t\| \\
\leq Q\|\hat{\bbPhi}_\bbS(\hbz_t)-\hat{\bbPhi}_\tbS(\tilde{\hbz}_t)\|\|\bbA\|\|\bbx_t\|
\end{split}
\end{align}
and, adding and subtracting $\hat{\bbPhi}_\tbS(\hbz_t)$ to $\|\hat{\bbPhi}_\bbS(\hbz_t)-\hat{\bbPhi}_\tbS(\tilde{\hbz}_t)\|$ and applying the triangle inequality,
\begin{align} \label{eqn:1.1.1}
\begin{split}
\|\hat{\bbPhi}_\bbS&(\hbz_t)-\hat{\bbPhi}_\tbS(\tilde{\hbz}_t)\| \leq \ldots\\
&\leq \|\hat{\bbPhi}_\bbS(\hbz_t)-\hat{\bbPhi}_\tbS(\hbz_t)\| + \|\hat{\bbPhi}_\tbS(\hbz_t)-\hat{\bbPhi}_\tbS(\tilde{\hbz}_t)\| \\
&\leq \phi_2 \varepsilon \|\hbz_t\| + \phi_1 \|\hbz_t -\tilde{\hbz}_t\|
\end{split}
\end{align}
where the second inequality follows from AS\ref{as6} and AS\ref{as7}.

To bound the second term on the right-hand side of \eqref{eqn:thm2_state_diff1.1}, we use the fact that the gate operator is additive. Explicitly,
\begin{align} \label{eqn:1.2}
\begin{split}
\|\tilde{\hcalQ}(\bbA\bbx_t)-\tilde{\hcalQ}(\tbA\bbx_t)\| &= \|\tilde{\hcalQ}(\bbA\bbx_t-\tbA\bbx_t)\| \\
&\leq \|\bbA-\tbA\|\|\bbx_t\|
\end{split}
\end{align}
where we have used the Cauchy-Schwarz inequality and the fact that $\|\tilde{\hcalQ}\|\leq 1$. Putting together equations \eqref{eqn:1.1}, \eqref{eqn:1.1.1} and \eqref{eqn:1.2}, we arrive at a bound for $\|\hcalQ(\bbA\bbx_t)-\tilde{\hcalQ}(\tbA\bbx_t)\|$,
\begin{align} \label{eqn:1}
\begin{split}
\|\hcalQ(\bbA\bbx_t)&-\tilde{\hcalQ}(\tbA\bbx_t)\|
\leq Q\phi_2\varepsilon\|\hbz_t\|\\
&+ Q\phi_1\|\hbz_t-\tilde{\hbz}_t\| + \|\bbA-\tbA\|
\end{split}
\end{align}
where we have additionally used Assumptions AS\ref{as1} and AS\ref{as4} to bound $\|\bbA\|$ and $\|\bbx_t\|$. 

We add and subtract $\tilde{\kcalQ}(\bbB\bbz_{t-1})$ to the right-hand side of \eqref{eqn:thm2_state_diff1} and use the triangle inequality to get 
\begin{align*} 
\begin{split}
\|\kcalQ(\bbB\bbz_{t-1})-&\tilde{\kcalQ}(\tbB\tbz_{t-1})\| \leq \ldots \\
&\leq \|\kcalQ(\bbB\bbz_{t-1})-\tilde{\kcalQ}(\bbB\bbz_{t-1})\| \\
&+ \|\tilde{\kcalQ}(\bbB\bbz_{t-1})-\tilde{\kcalQ}(\tbB\tbz_{t-1})\|.
\end{split}
\end{align*}
Using the Cauchy-Schwarz inequality and AS\ref{as5},
\begin{align} \label{eqn:2.1}
\begin{split}
\|\kcalQ(\bbB\bbz_{t-1})-&\tilde{\kcalQ}(\bbB\bbz_{t-1})\| \leq \ldots \\
&\leq \|\kcalQ-\tilde{\kcalQ}\|\|\bbB\|\|\bbz_{t-1}\| \\
&\leq Q\|\check{\bbPhi}_\bbS(\cbz_{t})-\check{\bbPhi}_\tbS(\tilde{\cbz}_{t})\|\|\bbB\|\|\bbz_{t-1}\|
\end{split}
\end{align}
and adding and subtracting $\check{\bbPhi}_\tbS(\cbz_t)$ in $\|\check{\bbPhi}_\bbS(\cbz_{t})-\check{\bbPhi}_\tbS(\tilde{\cbz}_{t})\|$,
\begin{align} \label{eqn:2.1.1}
\|\check{\bbPhi}_\bbS(\cbz_{t})-\check{\bbPhi}_\tbS(\tilde{\cbz}_{t})\| \leq \phi_2 \varepsilon \|\cbz_t\| + \phi_1 \|\cbz_t -\tilde{\cbz}_t\|.
\end{align}

To bound $\|\tilde{\kcalQ}(\bbB\bbz_{t-1})-\tilde{\kcalQ}(\tbB\tbz_{t-1})\|$, we use the fact that the gating operator is additive and bounded by $1$ to write
\begin{align} \label{eqn:2.2}
\begin{split}
\|\tilde{\kcalQ}(\bbB\bbz_{t-1})-\tilde{\kcalQ}(\tbB\tbz_{t-1})\| &= \|\tilde{\kcalQ}(\bbB\bbz_{t-1}-\tbB\tbz_{t-1})\| \\
&\leq \|\tilde{\kcalQ}\|\|\bbB\bbz_{t-1}-\tbB\tbz_{t-1}\| \\
&\leq \|\bbB\bbz_{t-1}-\tbB\tbz_{t-1}\|
\end{split}
\end{align}
where the first inequality follows from Cauchy-Schwarz. Finally, adding and subtracting $\tbB\bbz_{t-1}$ to $\|\bbB\bbz_{t-1}-\tbB\tbz_{t-1}\|$ and applying the triangle and Cauchy-Schwarz inequalities, we get
\begin{align} \label{eqn:2.2.1}
\|\bbB\bbz_{t-1}-\tbB\tbz_{t-1}\| \leq \|\bbB-\tbB\|\|\bbz_{t-1}\|+\|\tbB\|\|\bbz_{t-1}-\tbz_{t-1}\|.
\end{align}
Putting together equations \eqref{eqn:2.1} through \eqref{eqn:2.2.1}, we arrive at a bound for $\|\kcalQ(\bbB\bbz_{t-1})-\tilde{\kcalQ}(\tbB\tbz_{t-1})\|$,
\begin{align} \label{eqn:2}
\begin{split}
\|\kcalQ(&\bbB\bbz_{t-1})-\tilde{\kcalQ}(\tbB\tbz_{t-1})\| \leq Q\phi_2\varepsilon \|\cbz_t\|\|\bbz_{t-1}\|\\
&+Q\phi_1\|\cbz_t-\tilde{\cbz}_t\|\|\bbz_{t-1}\|+\|\bbB-\tbB\|\|\bbz_{t-1}\|\\
&+\|\bbz_{t-1}-\tbz_{t-1}\|.
\end{split}
\end{align}

Plugging \eqref{eqn:1} and \eqref{eqn:2} back in \eqref{eqn:thm2_state_diff1}, we then see that the upper bound for $\|\bbz_t-\tbz_t\|$ satisfies a recurrence relationship. Explicitly,
\begin{align} \label{eqn:thm2_recurrence}
\begin{split}
\|\bbz_t-\tbz_t\| &\leq \|\bbz_{t-1}-\tbz_{t-1}\| + Q\phi_2\varepsilon\|\hbz_t\| \\
&+ Q\phi_1\|\hbz_t-\tilde{\hbz}_t\|
+ \|\bbA-\tbA\| \\
&+  \|\bbz_{t-1}\|Q(\phi_2\varepsilon \|\cbz_t\|
+\phi_1\|\cbz_t-\tilde{\cbz}_t\|) \\
&+ \|\bbz_{t-1}\|\|\bbB-\tbB\|.
\end{split}
\end{align}
Note that, because the GRNN used to compute the input gate state $\hbz_t$ is not gated, we can use equations \eqref{eqn:bound_nongated_state} and \eqref{eqn:state_diff3} to bound the second and third terms on the right-hand side, which only depend on $\|\hbz_t\|$ and $\|\hbz_t-\tilde{\hbz}_t\|$. We also know bounds for $\|\cbz_t\|$ and $\|\cbz_t-\tilde{\cbz}_t\|$, but the last two terms of \eqref{eqn:thm2_recurrence} also depend on $\|\bbz_{t-1}\|$. To bound this term, we use \eqref{eqn:gatingGeneric} to write
\begin{align*}
\begin{split}
\|\bbz_i\| &= \|\sigma(\hcalQ(\bbA\bbx_i)+\kcalQ(\bbB\bbz_{i-1}))\| \\
&\leq \|\hcalQ(\bbA\bbx_i)+\kcalQ(\bbB\bbz_{i-1})\| \\
&\leq \|A\|\|\bbx_i\|+\|\bbB\|\|\bbz_{i-1}\|
\end{split}
\end{align*}
where the first inequality follows from AS\ref{as2} and the second from the Cauchy-Schwarz inequality and the fact that the gate operator norms are bounded by $1$. We conclude that the addition of gate operators has no effect on this bound and thus $\|\bbz_i\|$ can be bounded as in \eqref{eqn:bound_nongated_state}. 

Substituting \eqref{eqn:bound_nongated_state} and \eqref{eqn:state_diff3} in \eqref{eqn:thm2_recurrence}, we solve the recurrence for $\|\bbz_t-\tbz_t\|$ and use Lemma \ref{lemmaStab} to obtain
\begin{align} \label{eqn:thm2_state}
\begin{split}
\|\bbz_t-\tbz_t\| &\leq C'(1+\delta\sqrt{N})(t+t^2)\varepsilon \\
&+ Q\left(\phi_2+\phi_1C'(1+\delta\sqrt{N})\right)t^3\varepsilon\\
&+Q\phi_1C'(1+\delta\sqrt{N})t^4\varepsilon +\ccalO(\varepsilon^2)
\end{split}
\end{align}
where we have used $C'=\max\{C_{\bbA},C_{\bbB},C_{\hbA},C_{\hbB},C_{\cbA},C_{\cbB}\}$ and $\bbz_0 = \boldsymbol{0}$ [cf. AS\ref{as3}]. Plugging \eqref{eqn:thm2_state} and \eqref{eqn:state_diff3} into \eqref{eqn:thm2_output_diff1.1} and using Lemma \ref{lemmaStab} once again, we arrive at the theorem's main result,
\begin{align}
\begin{split}
\|\bby_t-\tby_t\| &\leq C(1+\delta\sqrt{N})(3t+t^2)\varepsilon \\
&+ Q\left(\phi_2+\phi_1C(1+\delta\sqrt{N})\right)t^3\varepsilon\\
&+Q\phi_1C(1+\delta\sqrt{N})t^4\varepsilon+\ccalO(\varepsilon^2)
\end{split}
\end{align}
where $C = \max\{C',C_{\bbC}\}$.

\end{proof}
}


\bibliographystyle{IEEEtran}
\bibliography{myIEEEabrv,bibGGRNN}

\begin{thebibliography}{10}
\providecommand{\url}[1]{#1}
\csname url@samestyle\endcsname
\providecommand{\newblock}{\relax}
\providecommand{\bibinfo}[2]{#2}
\providecommand{\BIBentrySTDinterwordspacing}{\spaceskip=0pt\relax}
\providecommand{\BIBentryALTinterwordstretchfactor}{4}
\providecommand{\BIBentryALTinterwordspacing}{\spaceskip=\fontdimen2\font plus
\BIBentryALTinterwordstretchfactor\fontdimen3\font minus
  \fontdimen4\font\relax}
\providecommand{\BIBforeignlanguage}[2]{{%
\expandafter\ifx\csname l@#1\endcsname\relax
\typeout{** WARNING: IEEEtran.bst: No hyphenation pattern has been}%
\typeout{** loaded for the language `#1'. Using the pattern for}%
\typeout{** the default language instead.}%
\else
\language=\csname l@#1\endcsname
\fi
#2}}
\providecommand{\BIBdecl}{\relax}
\BIBdecl

\bibitem{ruiz2019gated}
L.~Ruiz, F.~Gama, and A.~Ribeiro, ``Gated graph convolutional recurrent neural
  networks,'' in \emph{27th Eur. Signal Process. Conf.}\hskip 1em plus 0.5em
  minus 0.4em\relax Spain: IEEE, 2-6 Sep. 2019.

\bibitem{Kipf17-ClassifGCN}
T.~N. Kipf and M.~Welling, ``Semi-supervised classification with graph
  convolutional networks,'' in \emph{5th Int. Conf. Learning
  Representations}.\hskip 1em plus 0.5em minus 0.4em\relax Toulon, France:
  Assoc. Comput. Linguistics, 24-26 Apr. 2017, pp. 1--14.

\bibitem{Defferrard17-CNNGraphs}
M.~Defferrard, X.~Bresson, and P.~Vandergheynst, ``Convolutional neural
  networks on graphs with fast localized spectral filtering,'' in \emph{30th
  Conf. Neural Inform. Process. Syst.}\hskip 1em plus 0.5em minus 0.4em\relax
  Barcelona, Spain: Neural Inform. Process. Foundation, 5-10 Dec. 2016, pp.
  3844--3858.

\bibitem{Gama19-Architectures}
F.~Gama, A.~G.~Marques, G.~Leus, and A.~Ribeiro, ``Convolutional neural network
  architectures for signals supported on graphs,'' \emph{{IEEE} Trans. Signal
  Process.}, vol.~67, no.~4, pp. 1034--1049, Feb. 2019.

\bibitem{Ruiz20-Nonlinear}
L.~Ruiz, F.~Gama, A.~G.~Marques, and A.~Ribeiro, ``Invariance-preserving
  localized activation functions for graph neural networks,'' \emph{{IEEE}
  Trans. Signal Process.}, vol.~68, no.~1, pp. 127--141, Jan. 2020.

\bibitem{tolstaya2020learning}
E.~Tolstaya, F.~Gama, J.~Paulos, G.~Pappas, V.~Kumar, and A.~Ribeiro,
  ``Learning decentralized controllers for robot swarms with graph neural
  networks,'' in \emph{Conference on Robot Learning}, 2020, pp. 671--682.

\bibitem{Goodfellow16-DeepLearning}
I.~Goodfellow, Y.~Bengio, and A.~Courville, \emph{Deep Learning}, ser. The
  Adaptive Computation and Machine Learning Series.\hskip 1em plus 0.5em minus
  0.4em\relax Cambridge, MA: The {MIT} Press, 2016.

\bibitem{Gama19-Stability}
\BIBentryALTinterwordspacing
F.~Gama, J.~Bruna, and A.~Ribeiro, ``Stability properties of graph neural
  networks,'' \emph{arXiv:1905.04497v2 [cs.LG]}, 4 Sep. 2019. [Online].
  Available: \url{http://arxiv.org/abs/1905.04497}
\BIBentrySTDinterwordspacing

\bibitem{girault2015translation}
B.~Girault, P.~Gon{\c{c}}alves, and E.~Fleury, ``Translation and stationarity
  for graph signals,'' \emph{{E}cole {N}ormale {S}up{\'e}rieure de {L}yon,
  {I}nria {R}h{\^o}ne-{A}lpes, {R}esearch Report {RR}-8719}, Apr. 2015.

\bibitem{marques2017stationary}
A.~G. Marques, S.~Segarra, G.~Leus, and A.~Ribeiro, ``Stationary graph
  processes and spectral estimation,'' \emph{{IEEE} Trans. Signal Process.},
  vol.~65, no.~22, pp. 5911--5926, 2017.

\bibitem{perraudin2017stationary}
N.~Perraudin and P.~Vandergheynst, ``Stationary signal processing on graphs,''
  \emph{{IEEE} Trans. Signal Process.}, vol.~65, no.~13, pp. 3462--3477, July
  2017.

\bibitem{grassi2017time}
F.~Grassi, A.~Loukas, N.~Perraudin, and B.~Ricaud, ``A time-vertex signal
  processing framework: Scalable processing and meaningful representations for
  time-series on graphs,'' \emph{{IEEE} Trans. Signal Process.}, vol.~66,
  no.~3, pp. 817--829, Feb. 2018.

\bibitem{pascanu2013construct}
\BIBentryALTinterwordspacing
R.~Pascanu, C.~Gulcehre, K.~Cho, and Y.~Bengio, ``How to construct deep
  recurrent neural networks,'' \emph{arXiv:1312.6026 [cs.NE]}, 24 Apr. 2014.
  [Online]. Available: \url{http://arxiv.org/abs/1312.6026}
\BIBentrySTDinterwordspacing

\bibitem{graves2013generating}
\BIBentryALTinterwordspacing
A.~Graves, ``Generating sequences with recurrent neural networks,''
  \emph{arXiv:1308.0850 [cs.NE]}, 5 June 2014. [Online]. Available:
  \url{http://arxiv.org/abs/1308.0850}
\BIBentrySTDinterwordspacing

\bibitem{schuster1997bidirectional}
M.~Schuster and K.~K. Paliwal, ``Bidirectional recurrent neural networks,''
  \emph{{IEEE} Trans. Signal Process.}, vol.~45, no.~11, pp. 2673--2681, 1997.

\bibitem{seo2018structured}
Y.~Seo, M.~Defferrard, P.~Vandergheynst, and X.~Bresson, ``Structured sequence
  modeling with graph convolutional recurrent networks,'' in \emph{32nd Conf.
  Neural Inform. Process. Syst.}\hskip 1em plus 0.5em minus 0.4em\relax
  Montreal, QC: Springer, 3-8 Dec. 2018, pp. 362--373.

\bibitem{li2017diffusion}
Y.~Li, R.~Yu, C.~Shahabi, and Y.~Liu, ``Diffusion convolutional recurrent
  neural network: Data-driven traffic forecasting,'' in \emph{Int. Conf.
  Learning Representations 2018}.\hskip 1em plus 0.5em minus 0.4em\relax
  Vancouver, BC: Assoc. Comput. Linguistics, 30 Apr.-3 May 2018.

\bibitem{zhang2018gaan}
J.~Zhang, X.~Shi, J.~Xie, H.~Ma, I.~King, and D.-Y. Yeung, ``{GaAN}: Gated
  attention networks for learning on large and spatiotemporal graphs,'' in
  \emph{Conf. Uncertainty Artificial Intell. 2018}, no. 139.\hskip 1em plus
  0.5em minus 0.4em\relax Monterey, CA: Assoc. Uncertainty Artificial Intell.,
  6-10 Aug. 2018.

\bibitem{yu2017spatio}
B.~Yu, H.~Yin, and Z.~Zhu, ``Spatio-temporal graph convolutional networks: A
  deep learning framework for traffic forecasting,'' in \emph{27th Int. Joint
  Conf. Artificial Intell.}\hskip 1em plus 0.5em minus 0.4em\relax Stockholm,
  Sweden: Eur. Assoc. Artificial Intell., 13-19 July 2018, pp. 3634--3640.

\bibitem{Velickovic18-GraphAttentionNetworks}
P.~Veli{\v{c}}kovi{\'{c}}, G.~Cucurull, A.~Casanova, A.~Romero, P.~Li{\`{o}},
  and Y.~Bengio, ``Graph attention networks,'' in \emph{Int. Conf. Learning
  Representations 2018}.\hskip 1em plus 0.5em minus 0.4em\relax Vancouver, BC:
  Assoc. Comput. Linguistics, 30 Apr.-3 May 2018, pp. 1--12.

\bibitem{li2015gated}
\BIBentryALTinterwordspacing
Y.~Li, D.~Tarlow, M.~Brockschmidt, and R.~Zemel, ``Gated graph sequence neural
  networks,'' \emph{arXiv:1511.05493 [cs.LG]}, 22 Sep. 2017. [Online].
  Available: \url{http://arxiv.org/abs/1511.05493}
\BIBentrySTDinterwordspacing

\bibitem{ioannidis2018recurrent}
V.~N. Ioannidis, A.~G.~Marques, and G.~B. Giannakis, ``A recurrent graph neural
  network for multi-relational data,'' in \emph{44th {IEEE} Int. Conf. Acoust.,
  Speech and Signal Process.}\hskip 1em plus 0.5em minus 0.4em\relax Brighton,
  UK: IEEE, 12-17 May 2019.

\bibitem{wu2020evonet}
\BIBentryALTinterwordspacing
C.~Wu, G.~Nikolentzos, and M.~Vazirgiannis, ``{E}vo{N}et: A neural network for
  predicting the evolution of dynamic graphs,'' \emph{arXiv:2003.00842[cs.LG]},
  2020. [Online]. Available: \url{https://arxiv.org/abs/2003.00842}
\BIBentrySTDinterwordspacing

\bibitem{tolstaya19-flocking}
E.~Tolstaya, F.~Gama, J.~Paulos, G.~Pappas, V.~Kumar, and A.~Ribeiro,
  ``Learning decentralized controllers for robot swarms with graph neural
  networks,'' in \emph{Conf. Robot Learning 2019}.\hskip 1em plus 0.5em minus
  0.4em\relax Osaka, Japan: Int. Found. Robotics Res., 30 Oct.-1 Nov. 2019.

\bibitem{Li20-Planning}
Q.~Li, F.~Gama, A.~Ribeiro, and A.~Prorok, ``Graph neural networks for
  decentralized multi-robot path planning,'' in \emph{2020}.\hskip 1em plus
  0.5em minus 0.4em\relax Las Vegas, NV: IEEE, 25-29 Oct. 2020.

\bibitem{baziotis2017datastories}
C.~Baziotis, N.~Pelekis, and C.~Doulkeridis, ``{DataStories} at {SemEval-2017
  Task 4}: Deep {LSTM} with attention for message-level and topic-based
  sentiment analysis,'' in \emph{{Proceedings of the 11th International
  Workshop on Semantic Evaluation}}, Vancouver, BC, Aug. 2017, pp. 747--754.

\bibitem{miao2015eesen}
Y.~Miao, M.~Gowayyed, and F.~Metze, ``{EESEN}: {End-to-end} speech recognition
  using deep {RNN} models and {WFST}-based decoding,'' in \emph{IEEE Workshop
  on Automatic Speech Recognition and Understanding}.\hskip 1em plus 0.5em
  minus 0.4em\relax Scottsdale, AZ: IEEE, 13-17 Dec. 2015, pp. 167--174.

\bibitem{Sandryhaila13-DSPG}
A.~Sandryhaila and J.~M.~F. Moura, ``Discrete signal processing on graphs,''
  \emph{{IEEE} Trans. Signal Process.}, vol.~61, no.~7, pp. 1644--1656, Apr.
  2013.

\bibitem{Shuman13-SPG}
D.~I. Shuman, S.~K. Narang, P.~Frossard, A.~Ortega, and P.~Vandergheynst, ``The
  emerging field of signal processing on graphs: Extending high-dimensional
  data analysis to networks and other irregular domains,'' \emph{{IEEE} Signal
  Process. Mag.}, vol.~30, no.~3, pp. 83--98, May 2013.

\bibitem{Heimowitz17-MarkovGSP}
A.~Heimowitz and Y.~C. Eldar, ``A unified view of diffusion maps and signal
  processing on graphs,'' in \emph{2017 Int. Conf. Sampling Theory and
  Appl.}\hskip 1em plus 0.5em minus 0.4em\relax Tallin, Estonia: IEEE, 3-7 July
  2017, pp. 308--312.

\bibitem{Gama19-GraphConv}
F.~Gama, A.~G.~Marques, G.~Leus, and A.~Ribeiro, ``Convolutional graph neural
  networks,'' in \emph{53rd Asilomar Conf. Signals, Systems and Comput.}\hskip
  1em plus 0.5em minus 0.4em\relax Pacific Grove, CA: IEEE, 3-6 Nov. 2019.

\bibitem{du2018graph}
J.~Du, J.~Shi, S.~Kar, and J.~M.~F. Moura, ``On graph convolution for graph
  {CNNs},'' in \emph{2018 {IEEE} Data Sci. Workshop}.\hskip 1em plus 0.5em
  minus 0.4em\relax Lausanne, Switzerland: IEEE, 4-6 June 2018, pp. 239--243.

\bibitem{Segarra17-Linear}
S.~Segarra, A.~G.~Marques, and A.~Ribeiro, ``Optimal graph-filter design and
  applications to distributed linear network operators,'' \emph{{IEEE} Trans.
  Signal Process.}, vol.~65, no.~15, pp. 4117--4131, Aug. 2017.

\bibitem{Gama19-GLLN}
F.~Gama and A.~Ribeiro, ``Ergodicity in stationary graph processes: A weak law
  of large numbers,'' \emph{{IEEE} Trans. Signal Process.}, vol.~67, no.~10,
  pp. 2761--2774, May 2019.

\bibitem{Gama19-Control}
\BIBentryALTinterwordspacing
F.~Gama, E.~Isufi, A.~Ribeiro, and G.~Leus, ``Controllability of bandlimited
  graph processes over random time-varying graphs,'' \emph{arXiv:1904.10089v1
  [cs.SY]}, 22 Apr. 2019. [Online]. Available:
  \url{http://arxiv.org/abs/1904.10089}
\BIBentrySTDinterwordspacing

\bibitem{lecun15-deeplearning}
Y.~LeCun, Y.~Bengio, and G.~Hinton, ``Deep learning,'' \emph{Nature}, vol. 521,
  no. 7553, pp. 85--117, 2015.

\bibitem{kuo17-recos}
C.-C.~J. Kuo, ``The {CNN} as a guided multilayer {RECOS} transform,''
  \emph{{IEEE} Signal Process. Mag.}, vol.~34, no.~3, pp. 81--89, May 2017.

\bibitem{yosinski2014transferable}
J.~Yosinski, J.~Clune, Y.~Bengio, and H.~Lipson, ``How transferable are
  features in deep neural networks?'' in \emph{Conf. Neural Inform. Process.
  Syst.}, Montreal, QC, 8-13 Dec. 2014, pp. 3320--3328.

\bibitem{isufi19-edgenets}
E.~Isufi, F.~Gama, and A.~Ribeiro, ``Generalizing graph convolutional neural
  networks with edge-variant recursions on graphs,'' in \emph{27th Eur. Signal
  Process. Conf.}\hskip 1em plus 0.5em minus 0.4em\relax A Coru\~{n}a, Spain:
  Eur. Assoc. Signal Process., 2-6 Sep. 2019.

\bibitem{gama19-neurips}
F.~Gama, J.~Bruna, and A.~Ribeiro, ``Stability of graph scattering
  transforms,'' in \emph{33rd Conf. Neural Inform. Process. Syst.}\hskip 1em
  plus 0.5em minus 0.4em\relax Vancouver, BC: Neural Inform. Process. Syst.
  Foundation, 8-14 Dec. 2019.

\bibitem{Sandryhaila14-DSPGfreq}
A.~Sandyhaila and J.~M.~F. Moura, ``Discrete signal processing on graphs:
  Frequency analysis,'' \emph{{IEEE} Trans. Signal Process.}, vol.~62, no.~12,
  pp. 3042--3054, June 2014.

\bibitem{Daubechies92-Wavelets}
I.~Daubechies, \emph{Ten Lectures on Wavelets}, ser. CBMS-NSF Regional Conf.
  Series Appl. Math.\hskip 1em plus 0.5em minus 0.4em\relax Philadelphia, PA:
  SIAM, 1992, vol.~61.

\bibitem{pascanu2013difficulty}
R.~Pascanu, T.~Mikolov, and Y.~Bengio, ``On the difficulty of training
  recurrent neural networks,'' in \emph{Int. Conf. Mach. Learning}, Atlanta,
  GA, 16-21 June 2013, pp. 1310--1318.

\bibitem{bengio1994learning}
Y.~Bengio, P.~Simard, and P.~Frasconi, ``Learning long-term dependencies with
  gradient descent is difficult,'' \emph{IEEE Transactions on Neural Networks},
  vol.~5, no.~2, pp. 157--166, 1994.

\bibitem{susnjara2015accelerated}
A.~Susnjara, N.~Perraudin, D.~Kressner, and P.~Vandergheynst, ``Accelerated
  filtering on graphs using lanczos method,'' \emph{arXiv:1509.04537
  [math.NA]}, 2015.

\bibitem{liao2019lanczosnet}
R.~Liao, Z.~Zhao, R.~Urtasun, and R.~S. Zemel, ``Lanczosnet: {M}ulti-scale deep
  graph convolutional networks,'' \emph{arXiv:1901.01484 [cs.LG]}, 2019.

\bibitem{luan2019break}
S.~Luan, M.~Zhao, X.~Chang, and D.~Precup, ``Break the ceiling: Stronger
  multi-scale deep graph convolutional networks,'' in \emph{2019 Conf. Neural
  Inform. Process. Syst.}, 2019, pp. 10\,945--10\,955.

\bibitem{Gama18-NodeVariant}
F.~Gama, G.~Leus, A.~G.~Marques, and A.~Ribeiro, ``Convolutional neural
  networks via node-varying graph filters,'' in \emph{2018 {IEEE} Data Sci.
  Workshop}.\hskip 1em plus 0.5em minus 0.4em\relax Lausanne, Switzerland:
  IEEE, 4-6 June 2018, pp. 220--224.

\bibitem{contino2017distributed}
M.~Coutino, E.~Isufi, and G.~Leus, ``Advances in distributed graph filtering,''
  \emph{{IEEE} Trans. Signal Process.}, vol.~67, no.~9, pp. 2320--2333, May
  2019.

\bibitem{Isufi20-EdgeNets}
\BIBentryALTinterwordspacing
E.~Isufi, F.~Gama, and A.~Ribeiro, ``{E}dge{N}ets: Edge varying graph neural
  networks,'' \emph{arXiv:2001.07620v1 [cs.LG]}, 21 Jan. 2020. [Online].
  Available: \url{http://arxiv.org/abs/2001.07620}
\BIBentrySTDinterwordspacing

\bibitem{virmaux2018lipschitz}
A.~Virmaux and K.~Scaman, ``Lipschitz regularity of deep neural networks:
  analysis and efficient estimation,'' in \emph{32ndConf. Neural Inform.
  Process. Syst.}, 2018, pp. 3835--3844.

\bibitem{geonet}
{Earthquake Commission}, {GNS Science}, and {Land Information New Zealand},
  ``{GeoNet},'' https://www.geonet.org.nz/, 20 Feb. 2019.

\bibitem{jagadish2014big}
H.~V. Jagadish, J.~Gehrke, A.~Labrinidis, Y.~Papakonstantinou, J.~M. Patel,
  R.~Ramakrishnan, and C.~Shahabi, ``Big data and its technical challenges,''
  \emph{Comm. of the ACM}, vol.~57, no.~7, pp. 86--94, 2014.

\bibitem{kingma17-adam}
D.~P. Kingma and J.~L. Ba, ``{ADAM}: A method for stochastic optimization,'' in
  \emph{3rd Int. Conf. Learning Representations}.\hskip 1em plus 0.5em minus
  0.4em\relax San Diego, CA: Assoc. Comput. Linguistics, 7-9 May 2015.

\bibitem{mastrandrea2015contact}
R.~Mastrandrea, J.~Fournet, and A.~Barrat, ``Contact patterns in a high school:
  a comparison between data collected using wearable sensors, contact diaries
  and friendship surveys,'' \emph{PloS one}, vol.~10, no.~9, 2015.

\bibitem{fournet2017estimating}
J.~Fournet and A.~Barrat, ``Estimating the epidemic risk using non-uniformly
  sampled contact data,'' \emph{Scientific reports}, vol.~7, no.~1, pp. 1--14,
  2017.

\end{thebibliography}

\end{document}